\documentclass[letterpaper, 11pt]{article}
\pdfoutput=1
\usepackage[left=1in, right=1in, top=1in, bottom=1in]{geometry}
\usepackage[normalem]{ulem}
\usepackage{graphicx}
\usepackage{cite}
\usepackage{bm}
\usepackage{amsmath}
\usepackage{amssymb}
\usepackage{amsfonts}
\usepackage{amsthm}
\usepackage{cancel}
\usepackage{comment}
\usepackage{mathtools}
\usepackage[pagebackref]{hyperref}
\backrefparscanfalse
\renewcommand*{\backref}[1]{
\ifx\relax#1\relax
\else \!\!\!\!(#1) \fi}
\usepackage{array}
\usepackage{braket}
\usepackage{wrapfig}
\usepackage{dsfont}
\usepackage{tikz}
\usetikzlibrary{calc}
\usepackage{pgfplots}
\usepgfplotslibrary{fillbetween}
\usepackage[export]{adjustbox}
\usepackage{caption}
\usepackage{subcaption}

\usepackage{ifthen}
\usetikzlibrary{calc, decorations.pathmorphing,arrows.meta,backgrounds}

\hypersetup{colorlinks=true, urlcolor=[rgb]{0,0,0.5}, citecolor=[rgb]{0.5,0,0}, linkcolor=[rgb]{0,0,0.4}}
\pgfplotsset{compat=1.18}

\newtheorem{theorem}{Theorem}
\newenvironment{itheorem}[2]
{\vspace{0.5\baselineskip} \noindent \textbf{#1}~(#2)\textbf{.}~\itshape }{}
\newtheorem{definition}{Definition}
\newtheorem{lemma}{Lemma}

\newtheorem{proposition}{Proposition}

\newcommand{\AdS}{\mathrm{AdS}}

\def\autorefapp#1{\hyperref[#1]{Appendix~\ref{#1}}}

\renewcommand{\title}[1]{\vbox{\center\bf{\Large #1}}\vspace{5mm}}
\renewcommand{\author}[1]{\vbox{\center{#1}}\vspace{5mm}}
\newcommand{\address}[1]{\vbox{\center\em#1}}

\newcommand\emails[1]{\begingroup\renewcommand\thefootnote{}\footnote{#1}\addtocounter{footnote}{-1}\endgroup}

\def\vphi{\varphi}
\def\tr{{\rm tr}}

\def\R{\mathbb{R}}

\def\iden{\mathbb{I}}

\def\and{\quad {\rm and} \quad}

\def\ra{\rightarrow}
\def\nn{\nonumber\\}

\def\ketbra#1{ |{#1}\rangle\!\langle{#1}| }

\DeclareMathOperator*{\Ex}{\mathbb{E}}

\def\dist{\mathcal{D}}

\def\dA{D_A}
\def\da{D_a}
\def\dB{D_B}
\def\db{D_b}

\numberwithin{equation}{section}
\allowdisplaybreaks

\usepackage{qi_diagrams}

\def\ra{r_{A_1/A}}
\def\rb{r_{B_1/B}}
\def\rbb{r_{B_1\cup B_2/B}}
\def\rbm{r_{\cup_{i\neq m} B_i/B}}

\begin{document}

\title{Quantum Complexity Dynamics for Disjoint Subsystems}

\author{Yale Fan,${}^a$ Nicholas Hunter-Jones,${}^{b,c}$ Andreas Karch,${}^b$ Minju Kum,${}^b$ Shivan Mittal${}^b$}

\address{
${}^a$Department of Physics, University of Idaho, Moscow, ID 83844 \\[6pt]
${}^b$Department of Physics, University of Texas at Austin, Austin, TX 78712 \\[6pt]
${}^c$Department of Computer Science, University of Texas at Austin, Austin, TX 78712
}
\emails{\hspace*{-5mm} Emails: \href{mailto:yalefan@gmail.com}{\tt yalefan@gmail.com}, \{\href{mailto:nickrhj@utexas.edu}{\tt nickrhj},
\href{mailto:karcha@utexas.edu}{\tt karcha},
\href{mailto:minjukum@utexas.edu}{\tt minjukum},
\href{mailto:shivan@utexas.edu}{\tt shivan}\}{\tt @utexas.edu}.
}

\begin{abstract}
Quantum complexity has emerged as a natural probe of chaos, thermalization, and the black hole interior on timescales long after local observables have equilibrated. However, its time evolution has been studied almost exclusively in subsystems confined to a single connected region. We show, using complementary tools from holography and random quantum circuits, that noncontiguous subsystems composed of multiple disjoint regions give rise to qualitatively new physics compared to the contiguous case. First, at finite temperature, a subsystem occupying less than half of the total system can carry high complexity at late times, even as the complexity of its larger complement remains low. This inversion of the usual hierarchy is intrinsically thermal: it vanishes in the infinite-temperature limit, which is the regime modeled by random quantum circuits. Second, if a subsystem's complexity equilibrates at an early time, then fragmenting it into $m$ disjoint components can further reduce this timescale by a factor of $m$, a phenomenon we exhibit in both holography and random quantum circuits. These results not only sharpen the correspondence between geometric and computational notions of complexity, but also motivate the search for novel complexity phenomena in quantum dynamics.

\end{abstract}

\newpage

\tableofcontents

\newpage

\section{Introduction}

Holographic complexity is often motivated by its ability to track the time evolution of black hole interiors long after entanglement entropies have saturated \cite{Susskind:2014rva, Stanford:2014jda, Brown:2015bva, Brown:2015lvg, Carmi:2016wjl}.
For a thermofield double state, dual to a two-sided black hole geometry, the complexity of the full state grows with the volume of the Einstein--Rosen (ER) bridge for times that are exponentially long in the entropy $S$ \cite{Maldacena:2001kr, Hartman:2013qma, Stanford:2014jda}.
Restricting attention to a subsystem of the full state makes the story considerably richer. A reduced density matrix may look locally thermal, but the corresponding entanglement wedge may or may not contain a portion of the bridge. 
Therefore, the volume associated with that wedge, and hence the subsystem complexity, can change discontinuously with the topology of the Ryu--Takayanagi (RT) or Hubeny--Rangamani--Takayanagi (HRT) surface \cite{Ryu:2006bv, Ryu:2006ef, Hubeny:2007xt, Alishahiha:2015rta, Agon:2018zso}.

In previous work \cite{fan2025sharptransitionssubsystemcomplexity}, we exploited this basic phenomenon to identify several \emph{universal} sharp transitions in the time dependence of subsystem complexity driven by the reorganization of the entanglement wedge. 
These transitions are universal in the sense that they are expected to manifest in all chaotic quantum systems and not only in those with semiclassical holographic duals. 
We verified this expectation by proving that similar subsystem complexity transitions occur in random quantum circuits on $n$ qudits with local dimension $q$ such that $S \sim n \log q$, which capture generic features of chaotic quantum systems. Specifically, we demonstrated a sharp complexity transition as a function of subsystem size at exactly half the total system size.
Below half system size, the complexity of a subsystem peaks after a time at most polynomial in the number of qudits $n$, and then drops down to the low complexity of the maximally mixed state. 
Above half system size, subsystem complexity grows linearly for a time exponentially long in $n$, resulting in exponentially high complexities at late times. 
This behavior persists even in models that use alternative definitions of quantum complexity \cite{Haah:2025hyf}.

The results in \cite{fan2025sharptransitionssubsystemcomplexity, Haah:2025hyf} were obtained under the assumption that the subsystems are contiguous.
The goal of this work is to identify novel phenomena that arise when we allow the subsystems to be noncontiguous: that is, to consist of several disjoint regions. 
For contiguousness to be a meaningful property, the underlying system must possess some notion of locality, as both holographic theories and random quantum circuits do. 
Earlier studies of disjoint boundary subsystems \cite{Hubeny:2013gta, Ben-Ami:2014gsa, Abt:2017pmf} examined static ``geometric'' phase transitions in both one-sided and two-sided black hole backgrounds rather than time dependence, while the systematic study of complexity for disjoint regions is a nascent subject \cite{Caputa:2026ldd, Fujiki:2026ucr}. We will see that changing the topology (connectedness) and the geometry (size) of a subsystem can both have drastic effects on the corresponding complexity transitions in space and time. 
Unlike in our earlier work \cite{fan2025sharptransitionssubsystemcomplexity}, which allowed for arbitrary spacetime dimensions, here, we mostly limit ourselves to the special case of two-dimensional boundary theories with three-dimensional bulk duals, where calculations can be done most explicitly.\footnote{If necessary, we denote by $d$ and $d + 1$ the number of boundary and bulk spacetime dimensions, respectively.}

We present two main results for holographic subsystem complexity. 
Our first result is a particularly counterintuitive effect. 
Let $A$ and $B$ be complementary boundary subsystems of sizes $|A|$ and $|B|$, respectively, such that $|A|<|B|$.
Ordinary intuition, and every previously studied case, suggests that the smaller subsystem should be the simpler one: the late-time complexity of subsystem $A$ should at best be polynomial in the entropy $S$, while that of subsystem $B$ can be exponentially large in $S$. 
Holographically, one would expect the entanglement wedge of the larger subsystem to always include the ER bridge. 
This intuition fails for noncontiguous regions. 
By distributing $A$ into several pieces around the spatial circle, one can arrange for $A$ to have exponentially large complexity while $B$ retains small complexity. 
In the gravitationally dual bulk, this is reflected in the fact that the dominant RT surface places the ER bridge within the entanglement wedge of $A$ despite $A$ being the smaller subsystem. 
In the infinite-temperature limit, this effect disappears and the larger subsystem always has exponentially large complexity. 
Since random quantum circuits are only expected to capture this infinite-temperature limit, we expect them to be blind to the counterintuitive properties of the finite-temperature holographic system. 
Intuitively, random quantum circuit states resemble permutation-invariant (as an ensemble) Haar-random states at the cost of an error that diminishes with circuit depth. 
Therefore, the late-time complexities of contiguous and noncontiguous subsystems of equal total size are similar. 
Indeed, the complexity lower bounds of \cite{fan2025sharptransitionssubsystemcomplexity} continue to hold regardless of the topology of the subsystems.

Our second main result concerns the timescale at which complexity saturates for the subsystem with the smaller late-time complexity, which in the infinite-temperature limit is just the smaller subsystem. 
As in the contiguous case, the complexity rises and subsequently drops. 
The novelty in the noncontiguous case is that the complexity collapse time for a system composed of $m$ equal-size, disjoint intervals is shorter by a factor of $m$ compared to the contiguous case ($m=1$).
We derive this phenomenon in holographic systems as well as in random quantum circuits.

Furthermore, we holographically derive the quantum-mechanical prediction from \cite{fan2025sharptransitionssubsystemcomplexity} about the scale of the complexity collapse in rapidly equilibrating subsystems. Then, as a corollary to our second result, we show that in the case of $m$ equal-size, disjoint intervals, the maximum complexity attained before collapse is reduced by the same factor.

In \autoref{subsec:hol_setup}--\autoref{subsec:hol_six}, we derive the holographic subsystem complexity for subsystems with two and three disjoint intervals on each boundary in various temperature limits and for arbitrary ratios of subintervals. In \autoref{subsec:hol_msubi}, we study how to decrease the fractional size of a subsystem by increasing the number of subintervals $m$ while retaining exponential complexity. In \autoref{subsec:hol_time}, we derive our second holographic result about the thermalization time for subsystems composed of $m$ equal-size subintervals, which we later prove rigorously for random quantum circuits in \autoref{sec:qi}. In \autoref{sec:qtohol}, we holographically derive the order of the complexity collapse in rapidly equilibrating subsystems. In \autoref{sec:discuss}, we collect our observations and discuss open questions.

\section{Disjoint Subsystems in Holography} \label{sec:hol}

\subsection{Holographic Setup} \label{subsec:hol_setup}

Our model system is a two-dimensional CFT on a spatial circle of circumference $\ell$ in the thermofield double state at inverse temperature $\beta$, dual to an eternal BTZ black hole \cite{Banados:1992wn}.  Boundary subregions are taken to be unions of intervals on a single boundary, with the same subsystem chosen on both sides when discussing the time-dependent thermofield double geometry.\footnote{We argue that having identical subsystems on both sides is a good model for a random quantum circuit because it preserves the symmetry of the thermofield double state. The large-$N$ limit in holography corresponds to the limit of large local dimension $q$ in quantum mechanics; correspondingly, the complexity curves that we derive may be smooth at finite $N$ but become sharp at $N = \infty$.}  Thus each time slice of our two-sided holographic model consists of two circles that are divided into identical segments and connected by a wormhole.  We partition the circle into complementary subsystems $A$ and $B$:
\begin{equation}
A = \bigcup_{i=1}^m A_i, \qquad B = \bigcup_{i=1}^m B_i,
\end{equation}
where the intervals $A_i$ and $B_i$ alternate around the circle:
\begin{equation}
A_1, B_1, A_2, B_2, \ldots, A_m, B_m.
\end{equation}
The total number of intervals on each circle is therefore \(2m\), while \(m\) is the number of connected components in both \(A\) and \(B\). We let $p$ denote the fraction of the boundary occupied by subsystem $A$.

The entanglement entropy associated with a single interval of length $x$ in a CFT on the line at finite temperature---which is proportional to the length of the corresponding RT surface \cite{Ryu:2006bv, Ryu:2006ef}---is
\begin{equation}
f(x)\equiv \frac{c}{3}\log\left(\frac{\beta}{\pi\epsilon}\sinh\frac{\pi x}{\beta}\right),
\label{eq:f-def}
\end{equation}
where we have made the UV cutoff (lattice spacing) $\epsilon$ explicit and omitted a non-universal additive constant \cite{Calabrese:2004eu}. This formula continues to hold on the circle, with the only difference being that in the black hole phase, one must also allow for an RT surface going the other way around the black hole horizon.  In the cases below, this is implemented by replacing a potentially large interval length \(x\) by \(\min(x,\ell-x)\) when appropriate.

Following the ``Complexity = Volume'' (CV) proposal for subregions \cite{Alishahiha:2015rta, Agon:2018zso}, the complexity of subsystem $A$ is
\begin{equation}
\mathcal{C}_V(A) = \frac{1}{8\pi G_N L_{\AdS}}\max_{\Sigma_A}\operatorname{Vol}(\Sigma_A),
\label{eq:subregion-cv}
\end{equation}
where \(\Sigma_A\) is a bulk codimension-one slice in the entanglement wedge of \(A\) with boundary
\begin{equation}
\partial \Sigma_A = A\cup \Gamma_A
\end{equation}
and \(\Gamma_A\) is the RT/HRT surface homologous to \(A\).  For static RT surfaces outside the horizon, the maximal slice lies at constant Schwarzschild time and the volume can be computed directly by integrating over the spatial region bounded by \(A\) and \(\Gamma_A\).  For HRT surfaces crossing the bridge, this is no longer true because the surface \(\Sigma_A\) must be varied independently.

Before discussing the space of possibilities for interval sizes and arrangements, we give an illustrative example of the most striking prediction from holography: situations where the complexity of a subsystem with $p < 1/2$ exhibits exponentially long growth while that of its \emph{larger} complementary subsystem saturates quickly.  This effect occurs exclusively at finite temperature and for disjoint subsystems.  Let $B$ (the larger system) consist of one interval of length $\ell/2 - \delta$ and one interval of length $2\delta$, where $\delta$ is small but finite.  Let $A$ consist of two intervals of length $(\ell/2 - \delta)/2$.  For fixed $\beta$ and sufficiently small $\delta$, we have
\begin{equation}
\underbrace{\left(\sinh\frac{\pi(\ell/2 - \delta)}{2\beta}\right)^2}_{O(1)}\geq \underbrace{\sinh\frac{\pi(\ell/2 - \delta)}{\beta}}_{O(1)}\underbrace{\sinh\frac{2\pi\delta}{\beta}}_{O(\delta)}.
\label{fourintervalexample}
\end{equation}
In light of \eqref{eq:f-def}, this implies that there exists a range of $\delta$ for which the RT surface homologous to $B$ has smaller area than the RT surface homologous to $A$, despite that $B$ is larger.  In this range, the entanglement wedge of $A$ contains the black hole.  See \autoref{fig:B<1/2}.  This effect disappears in the limit of infinite temperature ($\beta\to 0$), where the $\sinh$ functions become exponentials and the complexity transition occurs strictly at $p = 1/2$.

In holography, it is always the case that if a subsystem exhibits exponentially long complexity growth, then its complement does not.  This is because, at late times, the dominant (minimal-area) RT surfaces are the disconnected ``caps'' that do not cross the horizon; the only question is whether the black hole (and hence the ER bridge with its exponentially long growth) lies on the ``inside'' or the ``outside'' of those caps.  Note also that if an RT surface is anchored to a boundary interval of size $> \ell/2$, then it automatically contains the black hole.  Therefore, a subsystem that contains an interval of size $> \ell/2$ always has exponential complexity, even if the area of the late-time RT surface homologous to it is smaller than that of its complement.  See \autoref{fig:B>1/2}.  This explains why, in the ``exotic'' example of \eqref{fourintervalexample}, we chose all intervals to have length $< \ell/2$.

\subsection{Four Intervals: \texorpdfstring{$m = 2$}{m = 2}} \label{subsec:hol_four}

\subsubsection{Setup and Definitions}

We start with the simplest case, $m = 2$.  We label the intervals sequentially by $A_1, B_1, A_2, B_2$, so that $A = A_1\cup A_2$ and $B = B_1\cup B_2$.  Since a configuration of $2m$ intervals on the circle is specified by $2m - 1$ parameters, our problem involves three free parameters, which we call $p, \ra, \rb$.  Without loss of generality, we make the following assumptions:\footnote{For simplicity, we only state the inequalities below as strict inequalities; the endpoints require a careful case-by-case study.}
\begin{itemize}
\item Let subsystem $A$ be smaller than subsystem $B$:
\begin{equation}
    \ell_A = p\ell, \qquad \ell_B = (1 - p)\ell, \qquad \ell_A + \ell_B = \ell,
\end{equation}
where $0 < p < 1/2$.
\item Let interval $A_1$ be smaller than interval $A_2$:
\begin{equation}
    \ell_{A_1} = \ra\ell_A, \qquad \ell_{A_2} = (1 - \ra)\ell_A, \qquad \ell_{A_1} + \ell_{A_2} = \ell_A,
\end{equation}
where $0 < \ra < 1/2$.
\item Let interval $B_1$ be smaller than interval $B_2$:
\begin{equation}
    \ell_{B_1} = \rb\ell_B, \qquad \ell_{B_2} = (1 - \rb)\ell_B, \qquad \ell_{B_1} + \ell_{B_2} = \ell_B,
\end{equation}
where $0 < \rb < 1/2$.
\end{itemize}
In summary:
\begin{equation}
    \ell_{A_1} = p\ra\ell, \qquad \ell_{A_2} = p(1 - \ra)\ell, \qquad \ell_{B_1} = (1 - p)\rb\ell, \qquad \ell_{B_2} = (1 - p)(1 - \rb)\ell.
\end{equation}
Only interval $B_2$ can have length greater than $\ell/2$.

In general, the entanglement entropy of a noncontiguous region in a 2D CFT is a complicated function of the interval sizes and locations \cite{Calabrese:2009ez}. In the case of a holographic CFT (i.e., a CFT with large central charge $c$), the entanglement entropy simplifies to a sum of single-interval terms of the form \eqref{eq:f-def} to leading order in $1/c$ \cite{Hubeny:2007re, Headrick:2010zt, Hartman:2013mia}. However, there are often many candidate configurations of RT surfaces \cite{Hubeny:2013gta, Ben-Ami:2014gsa, Abt:2017pmf}. In the $m=2$ case, the situation is still fairly simple in that there are only two candidates, namely the union of the RT surfaces that one would naturally associate with either $A_1$ and $A_2$ (``$A$ caps'') or $B_1$ and $B_2$ (``$B$ caps''). These two candidates are illustrated in Figures \ref{fig:B<1/2} and \ref{fig:B>1/2}. The ``correct'' RT surface is the one associated with the smaller entropy.

\begin{figure}[!htb]
\centering
\begin{subfigure}{0.4\textwidth}
    \centering
    \includegraphics[width=.8\linewidth]{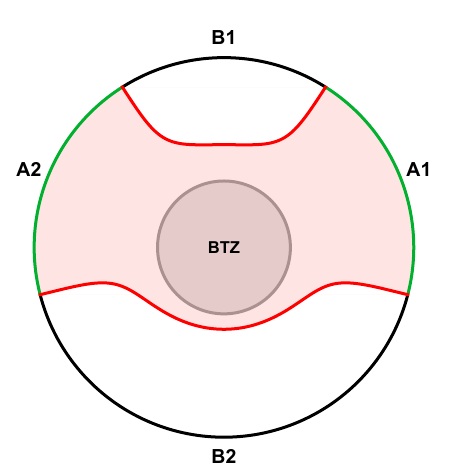}
\end{subfigure}
\begin{subfigure}{0.4\textwidth}
    \centering
    \includegraphics[width=.8\linewidth]{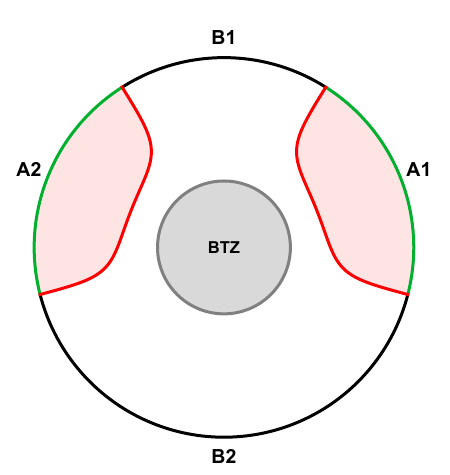}
\end{subfigure}
\caption{Two candidate RT surface configurations. All intervals have size $<\ell/2$. Left: the ``$B$ caps'' have smaller area, and the connected entanglement wedge of $A$ (colored in red) contains the wormhole despite that $\ell_A < \ell/2$. Right: the ``$A$ caps'' have smaller area, and the entanglement wedge of $A$ does not contain the wormhole.}
\label{fig:B<1/2}
\end{figure}

\begin{figure}[!htb]
\centering
\begin{subfigure}{0.4\textwidth}
    \centering
    \includegraphics[width=.8\linewidth]{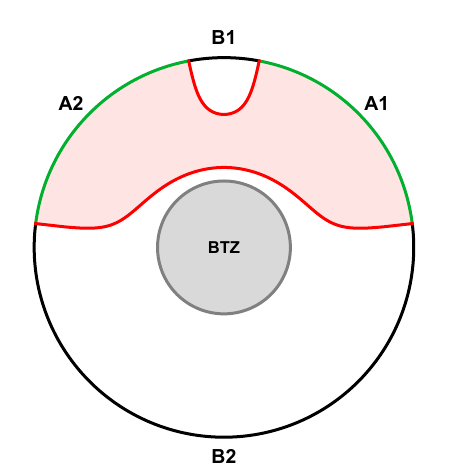}
\end{subfigure}
\begin{subfigure}{0.4\textwidth}
    \centering
    \includegraphics[width=.8\linewidth]{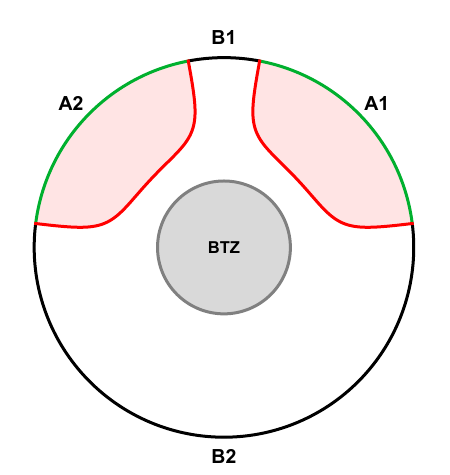}
\end{subfigure}
\caption{The case where the complementary subsystem $B$ contains an interval of size $> \ell/2$, namely $B_2$. Left: the ``$B$ caps'' have smaller area. Right: the ``$A$ caps'' have smaller area. In either case, the entanglement wedge of $A$ does not contain the wormhole.}
\label{fig:B>1/2}
\end{figure}

Let us denote the resulting entanglement entropies of the two candidate RT surfaces by
\begin{equation}\label{eq:SA-SB-def}
    S_A = f(\ell_{A_1}) + f(\ell_{A_2}), \qquad S_B = f(\ell_{B_1}) + f(\min(\ell_{B_2}, \ell - \ell_{B_2})),
\end{equation}
with $f(x)$ as in \eqref{eq:f-def}. To determine which of the two subsystems has the exponentially large late-time complexity, we need to answer two questions. First, which one is smaller, $S_A$ or $S_B$? This identifies the correct entanglement wedges. Second, which entanglement wedge contains the black hole and hence the ER bridge, leading to exponentially large complexity? In the particular case displayed in \autoref{fig:B>1/2}, the black hole lies in the entanglement wedge of $B$ irrespective of which RT surface wins. This happens whenever the length of $B_2$ is larger than $\ell/2$. So here, we always see the ``normal'' behavior that the larger system $B$ has the large complexity. The novel situation in which the smaller system $A$ has the large complexity arises when $S_A > S_B$ while $\ell_{B_2} < \ell/2$, as displayed in \autoref{fig:B<1/2}. In that case, the black hole with its ER bridge and exponentially large complexity lies in the entanglement wedge of the smaller system $A$.

So when is $S_A > S_B$?  The condition $S_A > S_B$ is equivalent to
\begin{equation}
    \sinh\frac{\pi\ell_{A_1}}{\beta}\sinh\frac{\pi\ell_{A_2}}{\beta} > \sinh\frac{\pi\ell_{B_1}}{\beta}\sinh\frac{\pi\min(\ell_{B_2}, \ell - \ell_{B_2})}{\beta}.
\label{thecondition}
\end{equation}
Before proceeding, note the following.  For any nonnegative real number $a$, let
\begin{equation}
    g_a(x)\equiv \sinh(ax)\sinh(a(1 - x)),
\end{equation}
which achieves a global maximum at $x = 1/2$:
\begin{equation}
    g_a(1/2) = \sinh^2(a/2).
\end{equation}
For any nonnegative real numbers $a, b$ with $a > b$, we have
\begin{equation}
    g_a(x)\geq g_b(x)
\label{gagb}
\end{equation}
for all $x\in [0, 1]$, with equality occurring at the endpoints $x = 0, 1$.  In fact, it is convenient to restrict the domain of $g_a(x)$ to $x\in [0, 1/2]$.  

We need to treat the cases $\ell_{B_2}\leq \ell/2$ and $\ell_{B_2} > \ell/2$ separately. In terms of $p$, $\ra$, $\rb$, we consider two cases:
\begin{itemize}
\item If $\ell_{B_2}\leq \ell/2$, then $S_A > S_B$ is equivalent to
\begin{equation}
\sinh\frac{p\ra\pi\ell}{\beta}\sinh\frac{p(1 - \ra)\pi\ell}{\beta} > \sinh\frac{(1 - p)\rb\pi\ell}{\beta}\sinh\frac{(1 - p)(1 - \rb)\pi\ell}{\beta},
\end{equation}
or
\begin{equation}
g_{p\pi\ell/\beta}(\ra) > g_{(1 - p)\pi\ell/\beta}(\rb).
\end{equation}
So we must take
\begin{equation}\label{eq:rbounds}
\frac{1 - 2p}{2(1 - p)}\leq \rb < g_{(1 - p)\pi\ell/\beta}^{-1}(g_{p\pi\ell/\beta}(\ra)),
\end{equation}
where the lower bound comes from $\ell_{B_2}\leq \ell/2$.  In the expression for the upper bound, the inverse exists by the property \eqref{gagb} and the restriction of the domain to $[0, 1/2]$.
\item If $\ell_{B_2}>\ell/2$, then $S_A > S_B$ is equivalent to
\begin{equation}
    \sinh\frac{p\ra\pi\ell}{\beta}\sinh\frac{p(1 - \ra)\pi\ell}{\beta} > \sinh\frac{(1 - p)\rb\pi\ell}{\beta}\sinh\frac{(p + (1 - p)\rb)\pi\ell}{\beta},
\end{equation}
or (noting that $g_a(x/a) = \sinh(x)\sinh(a - x)$ and extending the domain of $g_a$)
\begin{equation}
    g_{p\pi\ell/\beta}(\ra) > -g_{p\pi\ell/\beta}((1 - 1/p)\rb).
\end{equation}
By taking $\rb$ sufficiently small, we can always achieve this inequality.  We must also respect the inequality
\begin{equation}
    \rb < \frac{1 - 2p}{2(1 - p)},
\end{equation}
which we do anyway by taking $\rb$ sufficiently small.
\end{itemize}
To summarize, for fixed $p$ and $\ra$, we can \emph{always} achieve $S_A > S_B$ by taking $\rb$ sufficiently small.

Finally, we observe that in the infinite-temperature limit ($\beta\to 0$), the condition \eqref{thecondition} becomes impossible to satisfy, which corresponds to the standard situation where the larger system has the exponentially large complexity. This can also be seen directly from \eqref{eq:f-def}:
\begin{equation}
f(x) \xrightarrow{\beta \rightarrow 0} \frac{c}{3} 
\log \left ( \frac{\beta}{2 \pi \epsilon} e^{\pi x/\beta} \right ) \sim \frac{c \pi x}{3 \beta} + \text{constant}.
\end{equation}
That is, at high temperature, the candidate RT surface area scales linearly with the length of the interval, so the correct RT surface is \emph{always} the one in which the disconnected surfaces are caps on the smaller subsystem. This means that the larger system always has the connected entanglement wedge and hence contains the ER bridge with its exponentially large late-time complexity.

\subsubsection{Full Four-Interval Parameter Space}

Having established that it is possible for the smaller subsystem to have exponentially large complexity, we complete the study of the $m=2$ case by mapping out the full parameter space. Again, we assume that $A$ is the smaller system ($p<1/2$), and we focus on the ``interesting'' regime where $\ell_{B_2}\leq\ell/2$.

We have found that the range of $\rb$ ($=\ell_{B_1}/\ell_B$) for which $S_B<S_A$ is as follows:
\begin{equation}\label{eq:4int-bounds}
    \frac{1-2p}{2(1-p)}\leq \rb < g^{-1}_{(1-p)\pi \ell/\beta}(g_{p\pi\ell/\beta}(\ra)),
\end{equation}
where the lower bound comes from $\ell_{B_2}\leq\ell/2$ and the upper bound comes from $S_B<S_A$. In \autoref{fig:q=1/2}, we plot the lower and upper bounds for $\rb$ as functions of $p$ at different temperatures, setting $\ra=1/2$. We see that as the temperature increases, the allowed region shrinks, which is consistent with the fact that $S_B<S_A$ becomes impossible to satisfy in the infinite-temperature limit $\beta\rightarrow 0$. However, decreasing the temperature does not expand the allowed region all the way to the left. Rather, there exists a limiting upper bound in the low-temperature limit.

Also, notice that the upper bound in \eqref{eq:4int-bounds} is a monotonically increasing function of $\ra\leq1/2$. Thus, whatever the upper-bounding function (of $p$) is, $\ra$ will only affect the height of the function. Compare \autoref{fig:q=1/2} (where $\ra=1/2$) and \autoref{fig:q=1/10} (where $\ra=1/10$). We see that as $\ra$ decreases (i.e., as $\ell_{A_1}$ and $\ell_{A_2}$ become more asymmetric), the allowed regions shrink and the bounds become harder to satisfy.

\begin{figure}[!htb]
    \centering
    \includegraphics[width=0.8\linewidth]{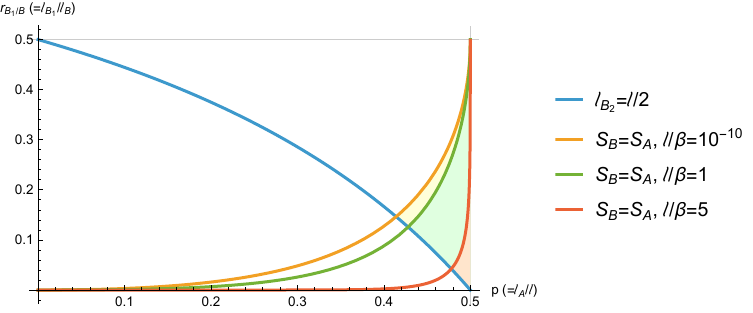}
    \caption{Lower and upper bounds for $\rb$ as functions of $p$, where $\ra=1/2$. The blue line is the lower bound. The yellow, green, and orange lines are the upper bounds for $\ell/\beta=10^{-10}$, $\ell/\beta=1$, and $\ell/\beta=5$, respectively. The allowed regions in the $(p,\rb)$ parameter space are colored accordingly. For example, when $\ell/\beta=1$, the allowed region is the union of the green and orange regions. As the temperature increases, the allowed region shrinks.}
    \label{fig:q=1/2}
\end{figure}

\begin{figure}[!htb]
    \centering
    \includegraphics[width=0.8\linewidth]{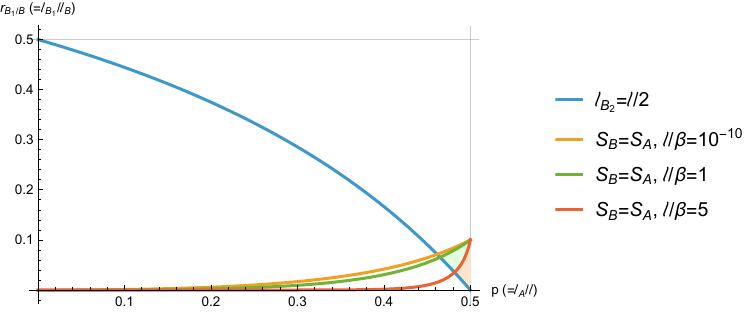}
    \caption{Same as \autoref{fig:q=1/2}, but with $\ra=1/10$. As $\ra$ decreases, the allowed region shrinks.}
    \label{fig:q=1/10}
\end{figure}

\subsubsection{Analytics at Limiting Temperatures}
\label{subsec:analyticstemp}

At limiting temperatures, we can obtain analytic expressions for the upper bounds. First, let us examine the low-temperature limit.\footnote{This limit is only relevant in $(2+1)$D, where black hole solutions exist all the way down to zero temperature (although they might not be thermodynamically preferred). In higher dimensions, there exists a minimum temperature (depending on $d$) below which black hole solutions do not exist \cite{Hawking:1982dh}. Below the Hawking-Page temperature, the bulk is given by two decoupled AdS spacetimes, so the complexity is always low even for the whole system (there is no ER bridge that could grow).} To do this, recall that in \autoref{fig:B<1/2}, the two candidate entropies are
\begin{align}
S_A &= \frac{c}{3}\left[\log\left(\frac{\beta}{\pi\epsilon}\sinh\frac{\pi\ell_{A_1}}{\beta}\right)+\log\left(\frac{\beta}{\pi\epsilon}\sinh\frac{\pi\ell_{A_2}}{\beta}\right)\right], \label{eq:4int-entropies-A} \\
S_B &= \frac{c}{3}\left[\log\left(\frac{\beta}{\pi\epsilon}\sinh\frac{\pi\ell_{B_1}}{\beta}\right)+\log\left(\frac{\beta}{\pi\epsilon}\sinh\frac{\pi\ell_{B_2}}{\beta}\right)\right]. \label{eq:4int-entropies-B}
\end{align}
In the low-temperature limit $\ell/\beta\rightarrow 0$, we can approximate $\sinh x\to x$, so
\begin{align}
    S_A\approx \frac{c}{3}\log\left(\frac{\ell_{A_1}\ell_{A_2}}{\epsilon^2}\right),\\
    S_B\approx \frac{c}{3}\log\left(\frac{\ell_{B_1}\ell_{B_2}}{\epsilon^2}\right).
\end{align}
The upper bound for $\rb$ is given by $S_A=S_B$, which reduces to
\begin{equation}
    \ell_{A_1}\ell_{A_2} =\ell_{B_1}\ell_{B_2}.
\end{equation}
Specializing to $\ra=1/2$ (that is, $A$ is evenly split), the relevant solution for $\rb<1/2$ is
\begin{equation}\label{eq:lowsol}
    \rb=\frac{1}{2}\left(1-\frac{\sqrt{1-2p}}{1-p}\right).
\end{equation}
We plot this solution in \autoref{fig:4interval-lowT}.

\begin{figure}[!htb]
    \centering
    \includegraphics[width=0.75\linewidth]{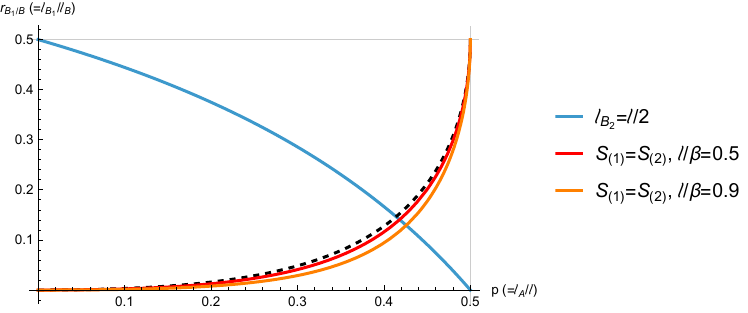}
    \caption{The upper bound for $\rb$ in the low-temperature limit \eqref{eq:lowsol} is marked as a black dashed curve. The upper bound approaches this curve as we lower the temperature.}
    \label{fig:4interval-lowT}
\end{figure}

Now we examine the high-temperature limit. In the high-temperature limit $\ell/\beta\rightarrow \infty$, we have $\sinh x\rightarrow\frac{1}{2}e^x$, so the entropies \eqref{eq:4int-entropies-A} and \eqref{eq:4int-entropies-B} become
\begin{align}
    S_A&\approx\frac{2c}{3}\log\frac{\beta}{2\pi\epsilon}+\frac{c\pi}{3\beta}\left(\ell_{A_1}+\ell_{A_2}\right),\\
    S_B&\approx\frac{2c}{3}\log\frac{\beta}{2\pi\epsilon}+\frac{c\pi}{3\beta}\left(\ell_{B_1}+\ell_{B_2}\right).
\end{align}
Equating $S_A=S_B$ reduces to
\begin{equation}
    \ell_{A_1}+\ell_{A_2}=\ell_{B_1}+\ell_{B_2},
\end{equation}
the solution for which is
\begin{equation}\label{eq:highsol}
    p=\frac{1}{2}.
\end{equation}
We plot this in \autoref{fig:4interval-highT}.

\begin{figure}[!htb]
    \centering
    \includegraphics[width=0.75\linewidth]{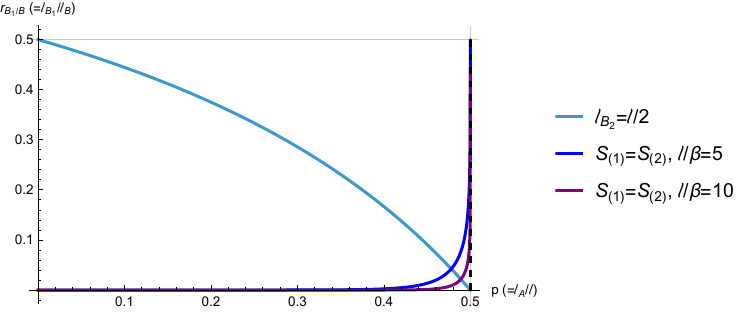}
    \caption{The ``upper bound'' for $\rb$ in the high-temperature limit \eqref{eq:highsol} is marked as a black dashed line. The upper bound approaches this line as we raise the temperature.}
    \label{fig:4interval-highT}
\end{figure}

Finally, consider perturbation around the infinite-temperature limit. Ideally, we would solve the equation $S_A = S_B$ for $\rb$ and obtain an analytic expression for the upper bound, but this is not very useful. Instead, let us read off how the minimum $p$ behaves with temperature $\ell/\beta$. That is, we look at the intersection point of the upper bound and the lower bound. We do this for the symmetric case \(\ra=1/2\). The endpoint of the allowed region occurs when the lower bound \(\ell_{B_2}=\ell/2\) is saturated together with \(S_A=S_B\). Thus
\begin{equation}
    \ell_{B_2}=\frac{\ell}{2},\qquad
    \ell_{B_1}=\left(\frac12-p\right)\ell,\qquad
    \ell_{A_1}=\ell_{A_2}=\frac{p\ell}{2}.
\end{equation}
Writing \(\alpha=\pi\ell/\beta\), the equation \(S_A=S_B\) becomes
\begin{equation}
\sinh^2\left(\frac{\alpha p}{2}\right)
=
\sinh\left[\alpha\left(\frac12-p\right)\right]\sinh\left(\frac{\alpha}{2}\right),
\end{equation}
or, equivalently,
\begin{equation}
2\cosh(\alpha p)-1=\cosh[\alpha(1-p)].
\end{equation}
This gives
\begin{equation}
p_\ast(\alpha)
=
\frac{1}{\alpha}
\log\left[
\frac{e^\alpha+(e^\alpha-1)\sqrt{2e^\alpha}}{2e^\alpha-1}
\right].
\end{equation}
In the high-temperature limit $\alpha\rightarrow\infty$,
\begin{equation}
p_\ast
=
\frac12-\frac{\log 2}{2\alpha}
+\cdots
=
\frac12-\frac{\beta\log 2}{2\pi\ell}
+\cdots.
\end{equation}
Thus, in the four-interval case, the exotic region exists only at finite temperature. As \(\beta\to 0\), its lower endpoint approaches \(p=1/2\), and for any fixed \(p<1/2\), the inequality \(S_A>S_B\) eventually fails.

\subsection{Six Intervals: \texorpdfstring{$m = 3$}{m = 3}} \label{subsec:hol_six}

\subsubsection{Setup and Dominant RT Surface}

We now proceed to study features of the six-interval case. As seen in \autoref{fig:q=1/2}, in the $m=2$ case, the minimum possible $p$ ($=\ell_A/\ell$)---that is, the smallest possible size $A$ can have while still having exponentially large complexity---is obtained as the intersection of the lower bound and the upper bound. We expect that increasing the number of intervals should make it possible to decrease the minimum $p$ further. For six intervals, we find that this is indeed the case.

Comparing \autoref{fig:q=1/2} and \autoref{fig:q=1/10}, we see that the value of the minimum $p$ in the $m=2$ case is lower when the configuration is symmetric ($\ra=1/2$, i.e., $\ell_{A_1}=\ell_{A_2}$). Assuming this property continues to hold in the six-interval case, we focus on the similarly symmetric configuration, where $\ell_{B_1}=\ell_{B_2}\neq\ell_{B_3}$ and $\ell_{A_1}=\ell_{A_2}=\ell_{A_3}$:
\begin{gather}
\ell_A=p\ell,\qquad \ell_B=(1-p)\ell, \notag \\
\ell_{A_1}=\ell_{A_2}=\ell_{A_3}=\frac{1}{3}p\ell, \label{eq:ratios} \\
\ell_{B_1}=\ell_{B_2}=\frac{1}{2}(1-p)\rbb\ell,\qquad \ell_{B_3}=(1-p)(1-\rbb)\ell, \notag
\end{gather}
with $0\leq p\leq 1/2$ and $0\leq \rbb\leq 1$. We illustrate the possible minimal surface configurations in \autoref{fig:six}.\footnote{As the number of intervals increases, the number of extremal surface configurations increases combinatorially as the Catalan numbers (1, 2, 5, 14, 42, \ldots). For six intervals, we have five different extremal surfaces. Among these extremal surface configurations, only the minimal one contributes to the entropy and defines the late-time entanglement wedge.}

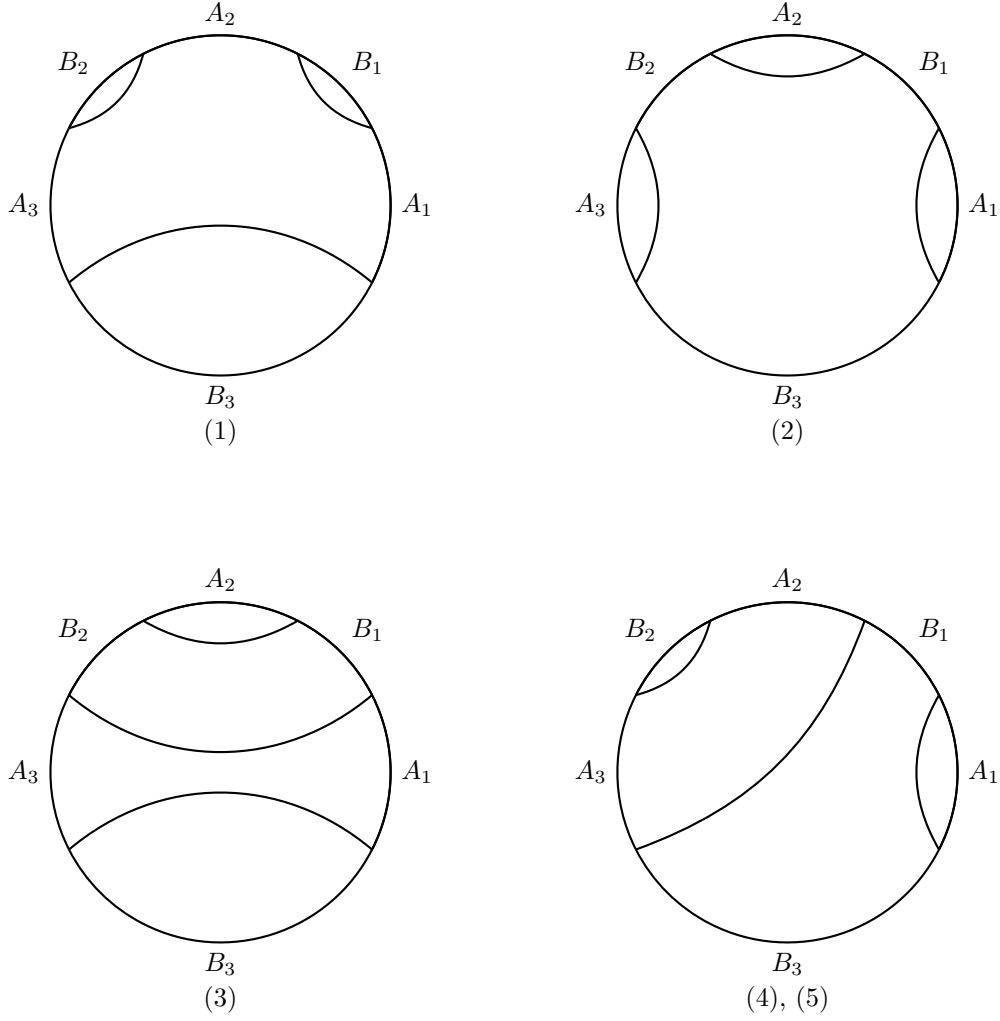
\begin{figure}
    \centering
\begin{tikzpicture}[
    scale=1.5,
    font=\small,
    thick,
    rt/.style={black, thick, bend right=35},
    boundary/.style={black, thick}
    ]

\def\angSone{27}
\def\angEone{333}
\def\angStwo{297}
\def\angEtwo{243}
\def\angSthree{207}
\def\angEthree{153}

\newcommand{\drawBaseCircle}[1]{

    \draw[boundary] (\angEthree:\R) arc (\angEthree:-27+360:\R);
    \draw (\angEthree:\R) arc (153:27:\R) node[midway, below] {$B_3$};
    \draw (27:\R) arc (27:-27:\R) node[midway, left] {$A_3$};
    \draw (333:\R) arc (333:297:\R) node[midway, above left] {$B_2$};
    \draw (297:\R) arc (297:243:\R) node[midway, above] {$A_2$};
    \draw (243:\R) arc (243:207:\R) node[midway, above right] {$B_1$};
    \draw (207:\R) arc (207:153:\R) node[midway, right] {$A_1$};

    \node at (0, \R+0.5) {#1};
}

\def\R{1.5}

\begin{scope}[shift={(0,0)},rotate=180]
    \drawBaseCircle{(1)}
    \draw[rt] (\angEthree:\R) to[bend right=40] (\angSone:\R);
    \draw[rt] (\angEone:\R) to[bend right=30] (\angStwo:\R);
    \draw[rt] (\angEtwo:\R) to[bend right=30] (\angSthree:\R);
\end{scope}

\begin{scope}[shift={(5,0)},rotate=180]
    \drawBaseCircle{(2)}
    \draw[rt] (\angSone:\R) to[bend right=30] (\angEone:\R);
    \draw[rt] (\angStwo:\R) to[bend right=30] (\angEtwo:\R);
    \draw[rt] (\angSthree:\R) to[bend right=30] (\angEthree:\R);
\end{scope}

\begin{scope}[shift={(0,-5)},rotate=180]
    \drawBaseCircle{(3)}
    \draw[rt] (\angEthree:\R) to[bend right=40] (\angSone:\R);
    \draw[rt] (\angEone:\R) to[bend right=40] (\angSthree:\R);
    \draw[rt] (\angEtwo:\R) to[bend left=30] (\angStwo:\R);
\end{scope}

\begin{scope}[shift={(5,-5)},rotate=180]
    \drawBaseCircle{(4), (5)}
    \draw[rt] (\angSone:\R) to[bend right=25] (\angEtwo:\R);
    \draw[rt] (\angEone:\R) to[bend right=30] (\angStwo:\R);
    \draw[rt] (\angSthree:\R) to[bend right=30] (\angEthree:\R);
\end{scope}

\end{tikzpicture}
\caption{Schematic illustration of the six-interval RT surface configurations. There are five different diagrams. The lengths are chosen so that $\ell_{B_1}=\ell_{B_2}\neq\ell_{B_3}$ and $\ell_{A_1}=\ell_{A_2}=\ell_{A_3}$. In this symmetric case, the fifth RT configuration is the left-right flipped version of the fourth one, hence the label ``(4), (5).'' The black hole is present but not explicitly drawn because it can lie in different wedges depending on the interval sizes.}
\label{fig:six}
\end{figure}

The various configurations have the following entanglement entropies:
\begin{align}
    S_{(1)}&=\frac{c}{3}\log\left(\frac{\beta}{\pi \epsilon}\sinh{\frac{\pi\ell_{B_1}}{\beta}}\right)+
    \frac{c}{3}\log\left(\frac{\beta}{\pi \epsilon}\sinh{\frac{\pi\ell_{B_2}}{\beta}}\right)+
    \frac{c}{3}\log\left(\frac{\beta}{\pi \epsilon}\sinh{\frac{\pi\ell_{B_3}}{\beta}}\right)\notag\\
    &=\frac{c}{3}\log\left(\sinh{\frac{\pi\ell_{B_1}}{\beta}}\sinh{\frac{\pi\ell_{B_2}}{\beta}}\sinh{\frac{\pi\ell_{B_3}}{\beta}}\right)+C,\notag\\
    S_{(2)}&=\frac{c}{3}\log\left(\sinh{\frac{\pi\ell_{A_1}}{\beta}}\sinh{\frac{\pi\ell_{A_2}}{\beta}}\sinh{\frac{\pi\ell_{A_3}}{\beta}}\right)+C,\label{eq:entropies}\\
    S_{(3)}&=\frac{c}{3}\log\left(\sinh{\frac{\pi\ell_{B_3}}{\beta}}\sinh{\frac{\pi\ell_{A_2}}{\beta}}\sinh{\frac{\pi(\ell_{B_1}+\ell_{A_2}+\ell_{B_2})}{\beta}}\right)+C,\notag\\
    S_{(4), (5)}&=\frac{c}{3}\log\left(\sinh{\frac{\pi\ell_{B_2}}{\beta}}\sinh{\frac{\pi\ell_{A_1}}{\beta}}\sinh{\frac{\pi(\ell_{A_2}+\ell_{B_2}+\ell_{A_3})}{\beta}}\right)+C,\notag
\end{align}
where $C=c\log \frac{\beta}{\pi\epsilon}$. Now there are three cases in which $A$ can contain the wormhole:
\begin{itemize}
    \item \textbf{Case 1: \autoref{fig:six}(1) is the minimal surface.}
    Two conditions should be satisfied:
    \begin{align}
        (i)&\quad \ell_{B_3}<\ell/2 \; \implies \; \tfrac{1-2p}{2(1-p)}<\rbb,\\
        (ii)&\quad S_{(1)} < S_{(2)}, S_{(3)}, S_{(4), (5)}.
    \end{align}

    Condition $(i)$ gives the same lower bound for $r$ as in the four-interval case. Condition $(ii)$ gives three corresponding bounds.

    \item \textbf{Case 2: \autoref{fig:six}(3) is the minimal surface.}
    Three conditions should be satisfied:
    \begin{align}
        (i)&\quad \ell_{B_3}<\ell/2 \; \implies\; \tfrac{1-2p}{2(1-p)}<\rbb,\\
        (ii)&\quad \ell_{B_1}+\ell_{A_2}+\ell_{B_2}<\ell/2\;\implies\; \rbb<\tfrac{1-2p/3}{2(1-p)},\\
        (iii)&\quad S_{(3)} < S_{(1)}, S_{(2)}, S_{(4), (5)}.
    \end{align}
    Conditions $(i)$ and $(ii)$ ensure that the $A_3$-$A_1$ wedge contains the wormhole.
    \item \textbf{Case 3: \autoref{fig:six}(4 or 5) is the minimal surface.} 
    Two conditions should be satisfied:
    \begin{align}
        (i)&\quad \ell_{B_1}+\ell_{A_1}+\ell_{B_3}<\ell/2\;\implies\; \tfrac{1-4p/3}{1-p}<\rbb,\\
        (ii)&\quad S_{(4), (5)} < S_{(1)}, S_{(2)}, S_{(3)}.
    \end{align}
\end{itemize}

It turns out that \textbf{Case 2} and \textbf{Case 3} are impossible. That is, it is impossible for (3), (4), or (5) to be the minimal surface \textit{while} $A$ has the wormhole and hence the exponentially large complexity. We prove this by showing that some lower bounds are always greater than some upper bounds when $A$ contains the wormhole.

Consider \textbf{Case 2}. The condition $S_{(3)}<S_{(4), (5)}$, for example, acts as a lower bound for $\rbb$:
\begin{equation}
    f_\text{lb}(p,\ell/\beta)<\rbb,
\end{equation}
where $f_\text{lb}(p,\ell/\beta)$ is some lower-bounding function. Now plug the upper-bounding function from $(ii)$ into $S_{(3)}$ and $S_{(4), (5)}$:
\begin{align}
S_{(3)}\Big\vert_{\rbb=\frac{1-2p/3}{2(1-p)}}&=\frac{c}{3}\log\left[\sinh\left(\frac{\pi \ell}{\beta}\frac{1}{2}\right)\sinh\left(\frac{\pi \ell}{\beta}\frac{p}{3}\right)\sinh\left(\frac{\pi \ell}{\beta}\frac{3-4p}{6}\right)\right],\notag\\
S_{(4), (5)}\Big\vert_{\rbb=\frac{1-2p/3}{2(1-p)}}&=\frac{c}{3}\log\left[\sinh\left(\frac{\pi \ell}{\beta}\frac{1-2p/3}{4}\right)\sinh\left(\frac{\pi \ell}{\beta}\frac{p}{3}\right)\sinh\left(\frac{\pi \ell}{\beta}\frac{1+2p}{4}\right)\right].
\end{align}
Then, since $\frac{1}{2}\geq\frac{1+2p}{4}$ and $\frac{3-4p}{6}\geq\frac{1-2p/3}{4}$ for all $p\leq 1/2$ (with inequality saturating at $p=1/2$), $S_{(3)}$ is always greater than or equal to $S_{(4), (5)}$. Hence we conclude that $f_\text{lb}(p,\ell/\beta)\geq \frac{1-2p/3}{2(1-p)}$. Since the lower bound is always greater than or equal to the upper bound, \textbf{Case 2} is impossible.
\textbf{Case 3} is proven similarly. The $S_{(4), (5)}<S_{(2)}$ condition acts as an upper bound for $\rbb$. By plugging the lower-bounding function from $(i)$ into $S_{(4), (5)}$ and $S_{(2)}$, one can similarly show that $S_{(4), (5)}>S_{(2)}$ for all $0\leq p\leq 1/2$. Since the lower bound is always greater than the upper bound, \textbf{Case 3} is impossible. Thus, in both \textbf{Case 2} and \textbf{Case 3}, the geometric condition for $A$ to contain the wormhole is incompatible with the corresponding entropy-minimization condition.

Therefore, we only need to look at \textbf{Case 1}. The bounds for low, medium, and high temperature are plotted in \autoref{fig:six-low}, \ref{fig:six-med}, \ref{fig:six-high}, respectively. Note that $r$ ranges from 0 to 1; in the four-interval case, we omitted the upper half due to the $r\leftrightarrow 1-r$ ($\ell_{B_1}\leftrightarrow \ell_{B_2}$) symmetry. The six-interval configuration does not have such an $r\leftrightarrow 1-r$ symmetry, so we need to consider the full range $0<r<1$. We notice several features:
\begin{itemize}
    \item As before, there exist limiting bounds in the zero-temperature limit. These are plotted in \autoref{fig:analytic}. Compare with \autoref{fig:six-low}.
    \item The $S_{(1)}<S_{(2)}$ bound (yellow) exhibits a turning point at $(p,\rbb)=(1/2,2/3)$ at all temperatures. This can be shown by equating $S_{(1)}$ and $S_{(2)}$ at $p=1/2$. This has a simple geometric interpretation: the point $(p,\rbb)=(1/2,2/3)$ is where all intervals have the same size, $\ell_{A_i}=\ell_{B_j}$ for all $i,j$. Hence the allowed region for $(p,\rbb)$ is divided into ``upper'' and ``lower'' regions by this turning point.\footnote{One may calculate this turning point for the general $2m$-interval case in a similarly symmetric setting, where $\ell_{A_1}=\cdots=\ell_{A_m}$, $\ell_{B_1}=\cdots=\ell_{B_{m-1}}\neq\ell_{B_m}$. One obtains $\rbm^\text{turn}=\frac{m-1}{m}$. This is the point where all $A$ and $B$ intervals have the same length $\ell/2m$.} See \autoref{fig:analytic}. Geometrically, the upper region has $\ell_{B_3}<\ell_{B_{1,2}}$, while the lower region has $\ell_{B_3}>\ell_{B_{1,2}}$.
    \item For all temperatures, we find that the minimum allowed $p$ has decreased compared to the four-interval case.
    \item In particular, while the minimum $p$ in the lower region approaches $1/2$ with increasing temperature (as in the four-interval case), the upper region behaves oppositely: it allows a very narrow region where $p\rightarrow 0$ as $T\rightarrow\infty$.
\end{itemize}

\begin{figure}[!htb]
\centering
\begin{subfigure}{.55\textwidth}
  \centering
  \includegraphics[width=\linewidth]{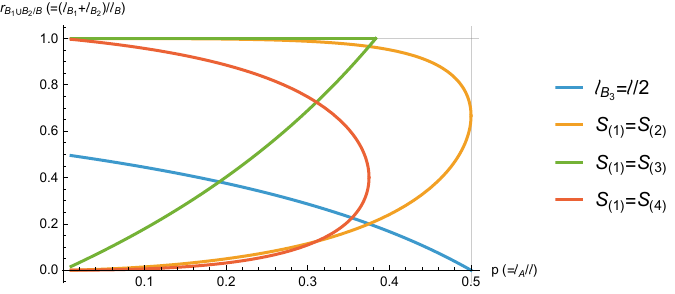}
  \caption{Low temperature, $\ell/\beta=10^{-10}$.}
  \label{fig:six-low}
\end{subfigure}%
\begin{subfigure}{.45\textwidth}
  \centering
  \includegraphics[width=0.95\linewidth]{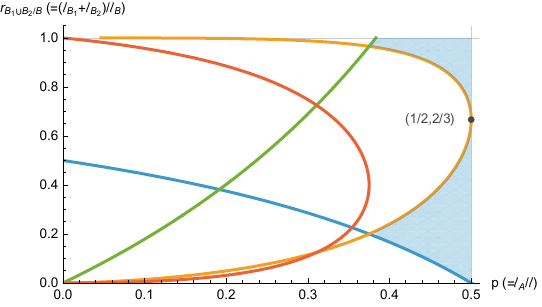}
  \caption{Analytic plot for zero-temperature limit.}
  \label{fig:analytic}
\end{subfigure}
\caption{Low-temperature case and analytic bounds in the zero-temperature limit. As in the four-interval case, the $T\rightarrow 0$ limit yields the limiting bounds. Analytic expressions can be found in \eqref{eq:low-analytic}. The region of interest, where \autoref{fig:six}(1) is the dominant surface and the smaller system $A$ has exponentially large complexity, is to the right of all the curves, as indicated in (b). Notice that the minimum allowed $p$ is smaller compared to the four-interval case (\autoref{fig:q=1/2}). }
\label{fig:six-low-analytic}
\end{figure}

\begin{figure}[!htb]
    \centering
    \begin{subfigure}{.55\textwidth}
        \centering
        \includegraphics[width=\linewidth]{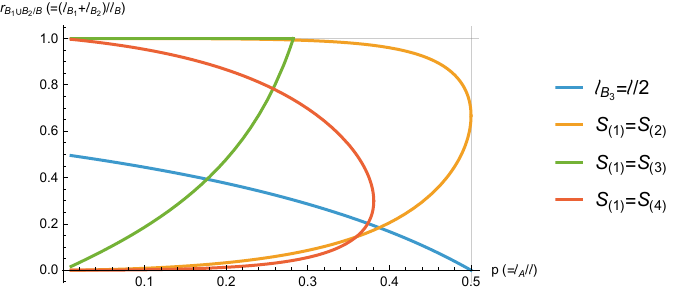}
        \caption{Medium temperature, $\ell/\beta=1$.}
        \label{fig:six-med}
    \end{subfigure}%
    \begin{subfigure}{.45\textwidth}
        \centering
        \includegraphics[width=0.95\linewidth]{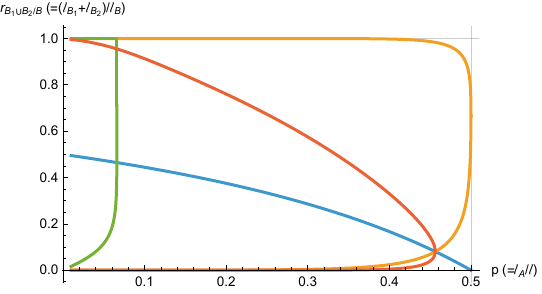}
    \caption{High temperature, $\ell/\beta=5$.}
    \label{fig:six-high}
    \end{subfigure}
    \caption{As temperature increases, smaller $p$ is allowed in the ``upper'' region ($\rbb>2/3$) because the $S_{(1)}<S_{(3)}$ bound (green), which yields the lower bound for $p$ in the upper region, becomes weaker. Meanwhile, the ``lower'' region ($\rbb<2/3$) shrinks with increasing temperature due to the $\ell_{B_3}<\ell/2$ bound (blue) and the $S_{(1)}<S_{(2)}$ bound (yellow).}
\end{figure}

\subsubsection{Analytics at Limiting Temperatures}\label{subsec:analyticssix}

As in the four-interval case, we can obtain analytic results at limiting temperatures.
Consider first the low-temperature limit ($\ell/\beta\ll 1$). In this limit, $\sinh\frac{\pi \ell}{\beta}\rightarrow\frac{\pi\ell}{\beta}$ and the temperature dependence drops out from the entropies; for example, the entropy of the first configuration $S_{(1)}$ becomes 
\begin{align}
    S_{(1)}&=\frac{c}{3}\log\left(\frac{\beta}{\pi \epsilon}\sinh{\frac{\pi\ell_{B_1}}{\beta}}\right)+
    \frac{c}{3}\log\left(\frac{\beta}{\pi \epsilon}\sinh{\frac{\pi\ell_{B_2}}{\beta}}\right)+
    \frac{c}{3}\log\left(\frac{\beta}{\pi \epsilon}\sinh{\frac{\pi\ell_{B_3}}{\beta}}\right)\notag\\
    &\approx\frac{c}{3}\log\left(\frac{\ell_{B_1}\ell_{B_2}\ell_{B_3}}{\epsilon^3}\right),
\end{align}
and comparing entropies amounts to comparing products of ratios. We can obtain the analytic expressions for the bounds in \autoref{fig:analytic} as functions of $r$ by solving $S_{(1)}=S_{(i)}$ for each $i=2,3,4$. These are the following (for notational simplicity, here we denote $r:=\rbb$):
\begin{align}\label{eq:low-analytic}
\ell_{B_3}=\ell/2:&\quad p=\frac{1-2r}{2(1-r)}, \notag \\
S_{(1)}=S_{(2)}:&\quad p
=
\frac{
3(r^2(1-r)/4)^{1/3}
}{
1+3(r^2(1-r)/4)^{1/3}
}, \notag \\
S_{(1)}=S_{(3)}:&\quad p=\frac{3(2r-2\sqrt{2}r+3r^2)}{-4 + 12 r + 9 r^2}, \notag \\
S_{(1)}=S_{(4)}:&\quad p=\frac{3(-7r+6r^2+\sqrt{16r-15r^2})}{2(3r-2)^2}.
\end{align}
Using these expressions, we can calculate where the bounds intersect in the $(p,\rbb)$ plane. In particular, we are interested in finding the minimum $p$. By direct calculation, we see that the $\ell_{B_3}=\ell/2$ bound and the $S_{(1)}=S_{(2)}$ bound intersect at $(p,\rbb)=(3/8,1/5)$, which gives the minimum value $p=0.375$ in the lower region, while the $S_{(1)}=S_{(2)}$ bound and the $S_{(1)}=S_{(3)}$ bound intersect at $(p,\rbb)=(3/8, 2(1+\sqrt{2})/5)$, which gives the same minimum value $p=0.375$ in the upper region.

The high-temperature limit is subtle because it reveals an interesting behavior at large but finite temperature, which disappears in the exact infinite-temperature limit. As before, set \(\alpha=\pi\ell/\beta\). If we keep \(p\) and \(\rbb\) fixed while taking \(\alpha\to\infty\), then
\begin{equation}
 \log\sinh(\alpha x)=\alpha x-\log 2+O(e^{-2\alpha x})
 \label{eq:highTlimit}
\end{equation}
and the entropy comparison reduces to comparing total lengths; therefore, the exotic phase disappears at sufficiently high temperature in this limit, just as in the four-interval case.

Nevertheless, the numerical plots show a narrow upper branch whose endpoint moves toward \(p=0\) as the temperature increases.  This is not a contradiction: the endpoint lies in a scaling regime in which \(p\) goes to zero while \(\rbb\) approaches one.  To capture it, take
\begin{equation}
 p=\frac{\xi}{\alpha},
 \qquad
 h=(1-p)(1-\rbb)\ll 1,
 \qquad
 \alpha\to\infty,
 \label{eq:highT_scaling}
\end{equation}
in such a way that \(\xi\) is held fixed and $h\alpha\ll 1$. The equation \(S_{(1)}=S_{(2)}\) fixes the exponentially small interval \(B_3\):
\begin{equation}
 h=\frac{4e^{-\alpha+\xi}}{\alpha}\sinh^3\!\left(\frac{\xi}{3}\right)
 \left[1+o(1)\right],
 \label{eq:h_scaling}
\end{equation}
with $o(1)$ error as a function of $\alpha$. The remaining nontrivial boundary is \(S_{(1)}=S_{(3)}\).  In the same scaling limit, this becomes
\begin{equation}
 \sinh\left(\frac{\xi}{3}\right)=\frac{1}{2} e^{-\xi/3} \; \implies \; e^{2\xi/3}=2 \; \implies \; \xi_*=\frac{3}{2}\log 2.
\end{equation}
Therefore, the lower endpoint of the upper branch scales as
\begin{equation}
 p_{\min}(T\to\infty)
 \sim
 \frac{\xi_*}{\alpha}
 =
 \frac{3\log2}{2}\frac{\beta}{\pi\ell}.
 \label{eq:highT_pmin}
\end{equation}
At the endpoint, \eqref{eq:h_scaling} reduces to
\begin{equation}
 h=(1-p)(1-\rbb)\sim \frac{e^{-\alpha}}{2\alpha},
\end{equation}
so the allowed region is squeezed exponentially close to \(r=1\).

Thus the high-temperature behavior of the six-interval case has two complementary descriptions.  For any fixed shape of the intervals, the exotic finite-temperature effect disappears as \(T\to\infty\).  However, if the shape is allowed to vary with temperature, then there remains a very narrow upper branch with
\begin{equation}\label{eq:squeezed}
 p_{\min}\propto \frac{\beta}{\ell},
 \qquad
 1-\rbb\propto e^{-\pi\ell/\beta},
\end{equation}
so that parametrically small subsystems can still contain the bridge at large but finite temperature.

Can we see this effect in the discrete qubit description on the boundary? The answer is no. In the discrete description, the boundary length $\ell$ corresponds to the number of qubits $n$ (multiplied by some dimensionful constant). Then the smallest fraction of an interval is $1/n$, which is proportional to $1/\ell$. Exponentially small fractions in $\ell\propto n$, like $e^{-\pi\ell/\beta}$ in \eqref{eq:squeezed}, are not realizable in qubits.

\subsection{Increasing \texorpdfstring{$m$}{m}} \label{subsec:hol_msubi}

We would like to know whether it is possible to decrease the minimum fraction $p_\text{min}$ occupied by subsystem $A$ (while allowing it to carry the exponentially large complexity) all the way to 0 by increasing the number of intervals $2m$. 

First, we consider a similarly symmetric configuration as before, $\ell_{A_1}=\cdots=\ell_{A_m}$ and $\ell_{B_1}=\cdots=\ell_{B_{m-1}}\neq\ell_{B_m}$. We focus on the low-temperature limit, where there are well-defined limiting bounds that give a $p_\text{min}$ that is also realizable in the boundary qubit description (i.e., no exponentially small fractions). One way to estimate $p_\text{min}$ is to see how the intersection points discussed in the $m=3$ case (namely, the intersection of the geometric lower bound with $S_{(1)}=S_{(2)}$ as well as the intersection of $S_{(1)}=S_{(2)}$ with $S_{(1)}=S_{(3)}$; see \autoref{fig:six-low-analytic}) behave in the large-$m$ limit. They might not provide the exact $p_\text{min}$ in this configuration, but they can serve as lower bounds for $p_\text{min}$.

\begin{figure}
    \centering
    \begin{tikzpicture}[scale=1.5, font=\small, thick]

\def\R{1.5}
\coordinate (center) at (0,0);

\def\angA{-55.8}
\def\angB{-23.4}
\def\angC{9}
\def\angD{41.4}
\def\angE{73.8}
\def\angF{106.2}
\def\angG{138.6}
\def\angH{171}
\def\angI{203.4}
\def\angJ{-124.2}

\draw (center) circle (\R);

\foreach \ang in {\angA, \angB, \angC, \angD, \angE, \angF, \angG, \angH, \angI, \angJ} {
    \draw (\ang:\R*0.97) -- (\ang:\R*1.03);
}

\node at (-90:\R+0.2) {$B_m$};
\node at (-39.6:\R+0.2) {$A_1$};
\node at (-7.2:\R+0.2) {$B_1$};
\node at (25.2:\R+0.2) {$A_2$};
\node at (57.6:\R+0.2) {$B_2$};
\node at (90:\R+0.2) {$A_3$};
\node at (122.4:\R+0.2) {$B_3$};
\node at (154.8:\R+0.2) {$\cdots$};
\node at (187.2:\R+0.4) {$B_{m-1}$};
\node at (219.6:\R+0.2) {$A_m$};

\begin{scope}[black, thick]
    \draw (\angJ:\R) to[bend left=35] (\angA:\R);
    \draw (\angB:\R) to[bend left=40] (\angE:\R);
    \draw (\angC:\R) to[bend left=25] (\angD:\R);
    \draw (\angF:\R) to[bend left=25] (\angG:\R);
    \draw (\angH:\R) to[bend left=25] (\angI:\R);
\end{scope}

\end{tikzpicture}
    \caption{The ``(3)" configuration in the $2m$-interval case.}
    \label{fig:2m-(3)}
\end{figure}
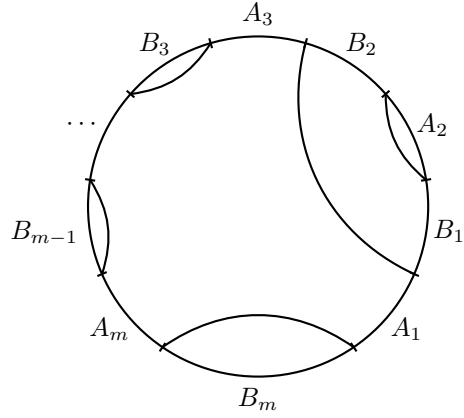

Let us look at the intersection point of the lower bound and the $S_{(1)}=S_{(2)}$ bound. It is given by the equations
\begin{equation}
    \log\ell_{B_1}\cdots\ell_{B_m}=\log\ell_{A_1}\cdots\ell_{A_m},\qquad\ell_{B_m}=\ell/2,
\end{equation}
which translates to
\begin{equation}
    \left(\frac{(1-p)\rbm}{m-1}\right)^{m-1}(1-p)(1-\rbm)=\left(\frac{p}{m}\right)^m,\quad (1-p)(1-\rbm)=\frac{1}{2}.
\end{equation}
In the $m\rightarrow\infty$ limit, one can show analytically that the solution asymptotes to $(p,r)\rightarrow(1/4,1/3)$.
Next, let us look at the intersection point of the $S_{(1)}=S_{(2)}$ bound and the $S_{(1)}=S_{(3)}$ bound. In the $2m$-interval case, the $S_{(3)}$ configuration is shown in \autoref{fig:2m-(3)}. The intersection point is given by the equations
\begin{equation}
    \log\ell_{B_1}\cdots\ell_{B_m}=\log\ell_{A_1}\cdots\ell_{A_m},\quad \log\ell_{B_1}\cdots\ell_{B_m}=\log{(\ell_{B_1}+\ell_{A_2}+\ell_{B_2})\ell_{A_2}\ell_{B_3}\cdots\ell_{B_m}},
\end{equation}
which translates to
\begin{gather}
    \left(\frac{(1-p)\rbm}{m-1}\right)^{m-1}(1-p)(1-\rbm)=\left(\frac{p}{m}\right)^m,\\ \left(\frac{(1-p)\rbm}{m-1}\right)^2=\left(\frac{2(1-p)\rbm}{m-1}+\frac{p}{m}\right)\frac{p}{m}.
\end{gather}
The solution for $\rbm$ is
\begin{equation}
    \rbm=\frac{(m-1)(1+\sqrt{2})^m}{1+(m-1)(1+\sqrt{2})^m},
\end{equation}
and in the $m\rightarrow\infty$ limit, the solution asymptotes to $(p,\rbm)\rightarrow(1 - \sqrt{2}/2,1)$.

These limiting intersection points show that with this particular symmetric choice of intervals, we cannot achieve arbitrarily small $p$ in the low-temperature limit regardless of the number of subintervals $m$. Instead, the minimum $p$ is lower-bounded by a nonzero value.

One may nevertheless expect that a large collection of intervals with almost vanishing total length, judiciously chosen, could yield an exponentially large complexity. A suggestion for how to accomplish this has been made in \cite{Pastawski:2016qrs} in the context of ``uberholography'': there, the authors argued that in the setting of empty AdS$_3$, one can choose a boundary subregion with fractal structure and negligible total length whose entanglement wedge nevertheless comprises almost the entire bulk.

For our purposes, the natural adaptation of their construction is to replace the inner logical boundary by the BTZ horizon. On a constant-time slice, the exterior region of a non-rotating BTZ black hole is an annulus with locally hyperbolic geometry,
\begin{equation}
 ds^2 = L^2\left(d\rho^2 + r_+^2\cosh^2\!\rho\, d\phi^2\right),
\qquad \rho\ge 0,
\end{equation}
where the horizon sits at $\rho=0$ and the asymptotic boundary at $\rho\to\infty$. Thus the geometric input behind their argument is essentially unchanged: one may start from a disconnected boundary region $A_0$ whose RT surface already places the horizon inside the entanglement wedge of $A_0$, and then recursively hollow out $A_0$ while preserving this property. The surviving set $A_{\tilde{m}}$ after $\tilde{m}$ iterations has total size
\begin{equation}
 |A_{\tilde{m}}| \sim \lambda^{\tilde{m}} |A_0|, \qquad 0<\lambda<1,
\end{equation}
so in the large-$\tilde{m}$ limit, one can make the total boundary fraction arbitrarily small while keeping the horizon inside the entanglement wedge. The resulting geometry is qualitatively sketched in \autoref{fig:fractalBTZ}.

Consequently, the exotic phase in which the entanglement wedge of the ``smaller'' subsystem crosses the bridge should not be thought of as a peculiarity of a few intervals; it should persist, and in fact we should expect it to become parametrically more dramatic, when one allows sufficiently intricate noncontiguous regions.

\begin{figure}[t]
\centering
\begin{tikzpicture}[scale=0.5]
  \definecolor{wedgeblue}{RGB}{156,225,226}

  \def\R{4.05}

  \begin{scope}
    \clip (0,0) circle (\R);
    \fill[wedgeblue] (0,0) circle (\R);

    \fill[white] (-5.15,0.15) circle (2.55);
    \fill[white] ( 5.10,0.10) circle (2.50);

    \fill[white] (-2.35,3.65) circle (0.63);
    \fill[white] (-1.30,3.93) circle (0.38);
    \fill[white] (0.10,4.28) circle (0.74);
    \fill[white] (1.28,3.95) circle (0.34);
    \fill[white] (2.18,3.74) circle (0.52);
    \fill[white] (2.85,3.78) circle (0.16);

    \fill[white] (-2.18,-3.74) circle (0.52);
    \fill[white] (-1.28,-3.95) circle (0.34);
    \fill[white] (-0.08,-4.28) circle (0.74);
    \fill[white] (1.35,-3.95) circle (0.38);
    \fill[white] (2.38,-3.68) circle (0.62);
    \fill[white] (-2.85,-3.78) circle (0.16);

    \draw[line width=1.1pt] (-5.15,0.15) circle (2.55);
    \draw[line width=1.1pt] ( 5.10,0.10) circle (2.50);

    \draw[line width=1.1pt] (-2.35,3.65) circle (0.63);
    \draw[line width=1.1pt] (-1.30,3.93) circle (0.38);
    \draw[line width=1.1pt] (0.10,4.28) circle (0.74);
    \draw[line width=1.1pt] (1.28,3.95) circle (0.34);
    \draw[line width=1.1pt] (2.18,3.74) circle (0.52);
    \draw[line width=1.1pt] (2.85,3.78) circle (0.16);

    \draw[line width=1.1pt] (-2.18,-3.74) circle (0.52);
    \draw[line width=1.1pt] (-1.28,-3.95) circle (0.34);
    \draw[line width=1.1pt] (-0.08,-4.28) circle (0.74);
    \draw[line width=1.1pt] (1.35,-3.95) circle (0.38);
    \draw[line width=1.1pt] (2.38,-3.68) circle (0.62);
    \draw[line width=1.1pt] (-2.85,-3.78) circle (0.16);
  \end{scope}
  \draw[line width=1.3pt] (0,0) circle (\R);
  \node[text=white] at (0,0) {\small BTZ};
  \fill[black] (0,0) circle (1.15);
  \node[text=white] at (0,0) {\small BTZ};
\end{tikzpicture}
\caption{A schematic adaptation of Figure 5 of \cite{Pastawski:2016qrs} to our finite-temperature setup. The blue region is the entanglement wedge of a highly disconnected boundary region. By recursively removing smaller boundary intervals, one can obtain a fractal-like boundary set of arbitrarily small total measure whose entanglement wedge still contains the BTZ horizon (shown in black at the center).}
\label{fig:fractalBTZ}
\end{figure}
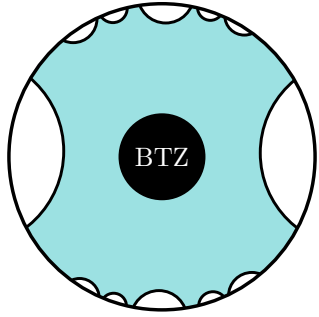

In $d = 2$, the key to this effect is the fact that entanglement entropy is a nonlinear function of boundary length at finite temperature.  An even richer phase diagram may exist in higher dimensions ($d > 2$), where one can arrange for a small boundary subsystem to reconstruct nearly the entire bulk: $B$ can be chosen so that its static minimal surface is a collection of caps and $A$ can be chosen so that it occupies a negligible fraction of the boundary.

\subsection{Thermalization Time} \label{subsec:hol_time}

We now look at the thermalization time for noncontiguous subsystems within holography. In \cite{fan2025sharptransitionssubsystemcomplexity}, we exploited the fact that the extremal surface of a subsystem undergoes a transition \cite{Hartman:2013qma} from a time-dependent Hartman-Maldacena (HM) surface to a static Ryu-Takayanagi (RT) surface at boundary time $t_b\sim\text{(subsystem size)}$. For subsystems of less than half the system size, this transition signaled an end to the linear growth of the complexity and a drop to the low complexity associated to the maximally mixed state.\footnote{In $d > 2$, for a single region with $p < 1/2$, holography would predict that the complexity saturation time scales as $t_b\sim p^{1/(d - 1)}\ell$ at infinite temperature, where $\ell$ is the linear size of the entire boundary.} This transition time is known as the thermalization time. How does this behavior change for noncontiguous subsystems?

If we only consider equal-size intervals where $\ell_{A_i}=\frac{\ell p}{m}$ and $\ell_{B_i}=\frac{\ell (1-p)}{m}$ for all $i=1,\ldots,m$ at infinite temperature $\ell/\beta\rightarrow\infty$, the answer is straightforward. At infinite temperature, the entropy scales linearly in the sum of the lengths of intervals (see, e.g., \eqref{eq:highTlimit}). Then there are only two possible late-time RT surfaces, those of type (1) and (2) in \autoref{fig:six}. Suppose our subsystem of interest is $A$ and $p<1/2$. Then the RT surface of type (2) is the dominant surface. Since each cap is independent, we can apply the argument of \cite{fan2025sharptransitionssubsystemcomplexity} to each of the intervals. We conclude that the transition happens at around
\begin{equation}\label{eq:transition-time-m}
    t_b\sim\frac{\ell_A}{m}.
\end{equation}

Now consider generalizing this to finite temperatures, while keeping the equal-size conditions. Notice that the RT surface of type (2) continues to be the smallest-area configuration at finite temperature. To see this, consider comparing the configurations of types (2) and (4), (5) in \autoref{fig:six}. The entropies corresponding to each of the configurations contain the following factors:
\begin{align}
    S_{(2)}&=\frac{c}{3}\log\left(\sinh\frac{\pi\ell_{A_1}}{\beta}\sinh\frac{\pi\ell_{A_2}}{\beta}\times(\text{common factor)}\right),\\
    S_{(4), (5)}&=\frac{c}{3}\log\left(\sinh\frac{\pi\ell_{B_1}}{\beta}\sinh\frac{\pi(\ell_{A_1}+\ell_{B_1}+\ell_{A_2})}{\beta}\times(\text{common factor)}\right),
\end{align}
where $\text{(common factor)}=\prod_{i=3}^m \sinh\frac{\pi\ell_{A_i}}{\beta}$. Assuming $p<1/2$ and the equal-size conditions, it is obvious that $\ell_{A_1}<\ell_{B_1}$ and $\ell_{A_2}<\ell_{A_1}+\ell_{B_1}+\ell_{A_2}$, so $S_{(2)}$ is always smaller. One can replace $S_{(4), (5)}$ with any type of configuration, apply the same argument, and conclude that $S_{(2)}$ is indeed the configuration of minimum area.
In \cite{fan2025sharptransitionssubsystemcomplexity}, it was shown that at finite temperature, there exists a nonzero $p_\text{crit}(\beta)$ such that a contiguous subsystem with $p<p_\text{crit}$ does not undergo the HM to RT transition. This is because, for subsystems with $p<p_\text{crit}$, RT surfaces are always smaller than HM surfaces. For subsystems with $p_\text{crit}<p<1/2$, the transition happens after a short time that grows sublinearly with $\ell_A$. Dividing a subsystem into $m$ equal-size intervals has the effect of increasing the critical fraction by a factor of $m$:
\begin{equation}
    p^{(m)}_\text{crit}=mp_\text{crit},
\end{equation}
where $p^{(m)}_\text{crit}$ denotes the critical fraction of a subsystem that consists of $m$ equal-size intervals. Therefore, if we divide the subsystem into sufficiently many intervals such that all intervals have a fraction smaller than $p_\text{crit}$, then the subsystem will not undergo the transition. For subsystems with $p^{(m)}_\text{crit}<p<1/2$, the transition time becomes the same sublinear function of $\ell_A/m$, simply because (2) continues to be the minimum-area configuration.

Next, consider generalizing to subsystems with different interval sizes while keeping the temperature infinite. Again, only configurations of type (1) or (2) dominate because the entropy depends on the sum of the lengths of the intervals (hence (2) dominates in the case of $p<1/2$). In this case, there will be multiple transitions in time, starting from the shortest interval, until the system fully thermalizes.

Finally, consider subsystems with different interval sizes at finite temperature. The minimal-area configuration is no longer restricted to either (1) or (2), but can be any one of the $C_m= \frac{1}{m+1}\binom{2m}{m}$ possible configurations (where $C_m$ is the $m^\text{th}$ Catalan number) depending on the temperature and relative sizes of the intervals. Once the minimal-area RT surfaces are fixed, the transition will proceed from the shortest interval (assuming it is big enough to undergo a transition) to the longest interval until the system fully thermalizes. One can imagine the following bizarre situation. Suppose $p<1/2$, and the minimal-area RT surfaces are such that $A$ includes the wormhole (e.g., \autoref{fig:six}(1)). Suppose all intervals associated with an RT surface are shorter than the critical length $\ell p_\text{crit}$. Then this configuration never undergoes a transition, and the larger system $B$ always has a fixed low complexity.

\section{Disjoint Subsystems of Random Quantum Circuits} \label{sec:qi}

Several infinite-temperature holographic predictions for quantum complexity have been rigorously proven in random quantum circuits. 
Ref.~\cite{brandao2021models} derived complexity growth for the circuit unitary and its pure-state output, rigorously formulating and establishing the Brown-Susskind conjecture \cite{Brown:2017jil} in quantum information theory. 
In \cite{fan2025sharptransitionssubsystemcomplexity}, we derived sharp transitions for holographic complexity in \emph{connected} subsystems and proved their (infinite-temperature) quantum information analogues in subsystems of random quantum circuits. In the previous sections, we established holographic predictions for complexity transitions and thermalization timescales for \emph{disjoint} subsystems; here, we rigorously prove the latter. In particular, we prove an $O(pn/m)$ thermalization timescale for a $pn$-size subsystem with a constant $p \in (0, 1/2)$ composed of $m$ equal-size, disjoint subintervals in the setting of an $n$-qudit 1D brickwork random quantum circuit with periodic boundary conditions. Furthermore, we observe that in the limit of large local dimension $q$, the subsystems thermalize suddenly: before the thermalization time, subsystems are maximally far in trace distance from the infinite-temperature thermal state. Combined with our complexity lower bound from \cite{fan2025sharptransitionssubsystemcomplexity}, which continues to hold for disjoint subsystems, our results from this section imply that the complexity of subsystems of random circuits grows linearly for $\Omega(pn/m)$ time to $O((pn)^2/m)$ before collapsing to the complexity of the maximally mixed state at $O(pn/m)$ time.

Before we present our results, we clarify our notation. 
Let $n$ (assumed to be even) denote the total number of $q$-dimensional qudits. 
Let $\nu_t$ denote the ensemble of unitaries in the unitary group $U(q^n)$ of depth-$t$ brickwork random quantum circuits with periodic boundary conditions, defined as follows.

\begin{definition}[Depth-$t$ random quantum circuits]
\label{def:bwrqc}
Depth-$t$ brickwork random quantum circuits with periodic boundary conditions on $n$ qudits are unitaries in $U(q^n)$ constructed by $t$-fold compositions of $U_{0,1}\otimes U_{2,3}\otimes \cdots \otimes U_{n-2, n-1}$ at odd and $U_{1,2}\otimes U_{3,4}\otimes \cdots \otimes U_{n-3,n-2} \otimes U_{n-1,0}$ at even iterations in the composition, where for $i, j \in \{0, 1, 2, \ldots, n-1\}$, each $U_{i, j}$ is a two-qudit unitary, drawn randomly with respect to the Haar measure $\mu$ on $U(q^2)$, that acts like the identity operator on all qudits not labeled by $i$ or $j$.
\end{definition}

We will write ``depth-$t$ random circuit'' to mean ``depth-$t$ brickwork random quantum circuit with periodic boundary conditions'' for brevity. 
The notation $U_t \sim \nu_t$ means that we draw the unitary $U_t \in U(q^n)$ from the ensemble $\nu_t$. 
$\Pr_{U_t \sim \nu_t} (E(U_t))$ denotes the probability of the event $E(U_t)$, a function on the random variable $U_t$ that is picked from $\nu_t$. 
We consider a bipartition of $n$ qudits into a subsystem $A$ and its complement $B$, each composed of $m$ disjoint parts. 
$n_A := pn$ and $n_B$ are the total sizes and $\dA$ and $\dB$ are the total Hilbert space dimensions of $A$ and $B$, respectively, where $p$ is the subsystem fraction. 
We consider the state $\rho_{A}(U_t) := \tr_{B}(U_t \ketbra{\vphi} U_t ^\dagger)$, where $\ket{\vphi} \in (\mathbb{C}^{q})^{\otimes n}$ is an arbitrary, fixed initial state and $U_t \sim \nu_t$. 
$\rho_{A}(U_t)$ is the state of subsystem $A$ of a depth-$t$ random circuit $U_t$. 
We use Schatten $p$-norms $\lVert X \rVert_p$ for $d \times d$ complex matrices $X$, defined as $\lVert X \rVert_p := (\sum_{i = 1} ^{d} s_i(X) ^p)^{1/p}$ where for $i \in \{1, 2, \ldots, d\}$, $s_i(X)$ denote the singular values of the matrix $X$. 
In particular, the trace distance $\mathcal{D}(\rho, \sigma)$ between two equal-dimensional quantum states $\rho$ and $\sigma$ (positive semidefinite and with $\tr(\rho) = \tr(\sigma) = 1$) is $\mathcal{D}(\rho, \sigma) = \frac{1}{2}\lVert \rho - \sigma \rVert_1$. In this section, for two functions $f, g$ of a single variable $x$, we write $f(x) = O(g(x))$ and $f(x) = \Omega(g(x))$ to mean that there exist constants $C, c > 0$ such that $f(x) \leq C g(x)$ and $f(x) \geq c g(x)$ for all sufficiently large $x$, respectively.
Our main result is the following:

\begin{itheorem}{Informal \autoref{thm:main}}{Thermalization timescale for disjoint subsystems}
Fix $\delta > 0$ and $p \in (0, 1/2)$. 
Consider subsystems $A$ of depth-$t$ random circuits $U_t \sim \nu_t$ of size $n_A = pn$ composed of $m$ equal-size, disjoint parts. 
The subsystem states $\rho_{A}(U_t) = \tr_{B}(U_t \ketbra{\vphi} U_t ^\dagger)$ are $\delta$-close in trace distance to the maximally mixed state, $\frac{\iden_A}{\dA}$, with high probability over $U_t \sim \nu_t$ after $t = O(n_A / 2m)$.
\end{itheorem}

\vspace{0.5\baselineskip}

We prove the above statement in two steps. First, we re-express the above probabilistic statement as the probability that $\mathcal{D}(\rho_{A}(U_t), \frac{\iden_A}{\dA}) \geq \delta$ is small, replace it with a weaker condition which is true more often using H\"older's inequality, and apply Markov's inequality:
\begin{align}
    \nonumber
    \Pr_{U_t \sim \nu_t} \Big(\mathcal{D}\Big(\rho_{A}(U_t), \frac{\iden_A}{\dA}\Big) \geq \delta \Big) &= \Pr_{U_t \sim \nu_t} \Big(\frac{1}{2}\Big\lVert \rho_{A}(U_t) - \frac{\iden_A}{\dA}\Big\rVert_1 \geq \delta \Big) && (\text{definition of trace distance})\\
    \nonumber
    &\leq \Pr_{U_t \sim \nu_t} \Big(\frac{\sqrt{\dA}}{2} \Big\lVert \rho_{A}(U_t) - \frac{\iden_A}{\dA}\Big\rVert_2 \geq \delta \Big) && (\text{H\"older's inequality}) \\
    \nonumber
    &= \Pr_{U_t \sim \nu_t} \Big(\dA \Big\lVert \rho_{A}(U_t) - \frac{\iden_A}{\dA}\Big\rVert_2 ^2 \geq 4\delta ^2 \Big) && (\text{square and rearrange}) \\
    \nonumber
    &= \Pr_{U_t \sim \nu_t} \Big( \dA \tr(\rho_{A} ^2(U_t)) - 1 \geq 4\delta ^2 \Big) \vphantom{\Big\lVert\Big\rVert_2^2} && (\text{evaluate Schatten $2$-norm}) \\
    &\leq \frac{\dA \Ex_{U_t \sim \nu_t} \tr(\rho_{A} ^2(U_t)) - 1}{4\delta^2} && (\text{Markov's inequality})\,.
\end{align}
Second, we evaluate the expected purity $\Ex_{U_t \sim \nu_t} \tr(\rho_{A} ^2(U_t))$ using a mapping to a random walk prob\-lem. The evolution of the purity of a subsystem has been well-studied in random quantum circuits \cite{ODP07,znidaric08,nahum2017quantum,cotler2022fluctuations,bensa22}. 
Under a statistical-mechanical mapping, the expected purity is a sum of contributions from $2m$ positively weighted, colliding random walkers. 
Our main technical contribution is \autoref{lem:main}, which is an upper bound on that sum that we achieve by further mapping the problem to one of non-intersecting chords on an annulus.

\subsection{Thermalization Timescale for Multiple Intervals in 1D} \label{subsec:qi_purity}

Consider a finite number $n$ of $q$-dimensional qudits that are partitioned into a subsystem $A$ and its complement. 
Consider a Haar-random gate $U \sim \mu(U(q^2))$ that acts nontrivially on a pair of qudits $i, j$ such that $i \in A$ and $j \not\in A$. 
We denote the states of $A$, $A - \{i\}$, and $A + \{j\}$ by $\rho_A$, $\rho_{A - \{i\}}$, and $\rho_{A + \{j\}}$ before the gate and by $\rho_A(U)$, $\rho_{A - \{i\}}(U)$, and $\rho_{A + \{j\}}(U)$ after the gate, respectively. The average purity of $\rho_{A}(U)$ over a Haar-random gate $U \sim \mu(U(q^2))$ is given by \cite{nahum2017quantum}
\begin{align}
    \Ex_{U \sim \mu(U(q^2))} \big[\tr(\rho_{A} ^2 (U))\big] = \frac{q}{q^2 + 1} \Big( \tr\big(\rho_{A-\{i\}}^2\big) + \tr\big(\rho_{A+\{j\}}^2\big)\Big)\,.
\end{align}
Beginning with a fully separable state $\ket{\vphi}$ on the $n$ qudits, the above relation can then be interpreted as a time-reversed, symmetric, Markovian random walk (i.e., the final circuit time corresponds to the initial random-walk time) \cite{cotler2022fluctuations}. 
For details, see the proof of Proposition 1 in \cite{cotler2022fluctuations} as well as Appendix A of \cite{Haah:2025hyf}. 
For that choice of $\ket{\vphi}$, and using the random walk point of view, the authors of \cite{cotler2022fluctuations} derived the average purity of a single contiguous $n_A$-size partition $A$ in an infinite 1D brickwork random quantum circuit with periodic boundary conditions:
\begin{equation}
    \Ex_{U_t \sim \nu_t} \big[\tr(\rho_A(U_t)^2) \big] = g_A(t) \left(\frac{q}{q^2+1} \right)^{2t} + \sum_{t'=1}^t c_A(t') \left(\frac{q}{q^2+1} \right)^{2t'}\,,
\end{equation}
where $c_A(t')$ is the number of intersections of two random walkers at a time step $t'$ and $g_A(t)$ is the number of configurations of two non-intersecting random walkers after $t$ steps, i.e., with no previous intersections. 
In particular,
\begin{equation}
c_A(t) = \frac{n_A}{2 t}\binom{2t}{t-n_A/2}\,.
\end{equation}
Using Lagrange inversion to evaluate the $t \rightarrow \infty$ limit of the sum $\sum_{t'=1}^t c_A(t') \left(\frac{q}{q^2+1} \right)^{2t'}$ and using the worst-case upper bound $g_A(t) \leq 2^{2t}$, we obtain
\begin{equation}
    \Ex_{U_t \sim \nu_t} \big[\tr(\rho_A(U_t)^2) \big] \leq \frac{1}{\dA} + \left(\frac{2q}{q^2+1} \right)^{2t}\,,
\end{equation}
where $\dA$ is the Hilbert space dimension of $A$. 
Thus we arrive at the following useful lemma for contiguous intervals, e.g., a single part of the subsystem $A$ or of its complement $B$.

\begin{lemma}
\label{lem:base}
For any contiguous interval $I$, the contribution to the average purity from that interval is bounded as
\begin{equation}
    \sum_{t'=1}^t c_I(t') \left(\frac{q}{q^2+1} \right)^{2t'} \leq \frac{1}{D_I}\,,
\end{equation}
where $c_I(t')$ is the number of intersections of two random walkers at a time step $t'$ and $d_I$ is the total Hilbert space dimension of the qudits in the contiguous interval $I$.
\end{lemma}

{\bf Two-interval case.} 
In the two-interval case, where the subsystem $A$ is contiguous (and its complement $B$ is not infinite), the purity decay is governed by two random walkers as before, who can now meet in two ways: by traversing either $n_A$ or $n_B$ qudits. 
In this case, the averaged purity decay is
\begin{equation}
    \Ex_{U_t \sim \nu_t} \big[\tr(\rho_A(U_t)^2) \big] \leq \frac{1}{\dA} + \frac{1}{\dB} + \left(\frac{2q}{q^2+1}\right)^{2(t-1)}\,,
\end{equation}
and most states $\rho_A(U_t)$ are close to $\iden_A/\dA$ at the timescale $t\sim n_A/2$.

{\bf Four-interval case.} 
Assume that the subsystem $A$ is comprised of two disjoint subintervals $A_1$ and $A_2$, separated by subintervals $B_1$ and $B_2$. 
In this case, the averaged purity decay is
\begin{align}
    \Ex_{U_t \sim \nu_t} \big[\tr(\rho_A(U_t)^2) \big] &\leq \frac{1}{D_{A_1}}\left(\frac{1}{D_{A_2}} + \frac{1}{D_{A_2 ^c}}\right) + \frac{1}{D_{B_1}}\left(\frac{1}{D_{B_2}} + \frac{1}{D_{B_2 ^c}}\right) + \frac{1}{D_{A_2}}\frac{1}{D_{A_1 ^c}} + \frac{1}{D_{B_2}}\frac{1}{D_{B_1 ^c}}\nn
    \label{eq:qi_four}
    &\qquad + \left(\frac{1}{D_{A_1}} + \frac{1}{D_{A_2}} + \frac{1}{D_{B_1}} + \frac{1}{D_{B_2}}\right) \left(\frac{2q}{q^2+1}\right)^{2(t-1)} + \left(\frac{2q}{q^2+1}\right)^{4(t-1)}\,,
\end{align}
where the superscript $c$ stands for ``complement.'' 
Then, for equal-sized $A_i$ subintervals, we find that the timescale at which most $\rho_A(U_t)$ are close to $\iden_A/\dA$ is $t\sim n_A/4$.

{\bf Six-interval case.} 
Assume that the subsystem $A$ is comprised of three disjoint intervals $A_1$, $A_2$, and $A_3$, separated by intervals $B_1$, $B_2$, and $B_3$. 
In this case, the averaged purity decay is
\begin{align}
    \nonumber
    &\hspace{-2mm} \Ex_{U_t \sim \nu_t}\big[\tr(\rho_A(U_t)^2) \big] \\
    \nonumber
    &\leq \frac{1}{D_{A_1}D_{A_2}D_{A_3}} + O\Big(\frac{1}{D_{B_i}}\Big) + \left(\frac{1}{D_{A_1}D_{A_2}} + \frac{1}{D_{A_2}D_{A_3}} + \frac{1}{D_{A_3}D_{A_1}} + O\Big(\frac{1}{D_{B_i}}\Big)\right) \left(\frac{2q}{q^2+1}\right)^{2(t-1)}\\
    \label{eq:qi_six}
    &+ \left(\frac{1}{D_{A_1}} + \frac{1}{D_{A_2}} + \frac{1}{D_{A_3}} + \frac{1}{D_{B_1}} + \frac{1}{D_{B_2}} + \frac{1}{D_{B_3}}\right) \left(\frac{2q}{q^2+1}\right)^{4(t-1)} + \left(\frac{2q}{q^2+1}\right)^{6(t-1)}\,.
\end{align}
Then, for equal-sized $A_i$ subintervals, we find that the timescale at which most $\rho_A(U_t)$ are close to $\iden_A/\dA$ is $t\sim n_A/6$.

{\bf Any number of intervals.} 
For any even number $2m$ of equal-sized, disjoint intervals (with $m$ subintervals in each of $A$ and $B$), the above examples suggest that we should find a thermalization timescale of $t\sim n_A/2m$. 
In the following, we rigorously prove this expectation. 
In particular, for $m = O(n_A)$, the thermalization timescale is a constant. In this scenario, the complexity grows and saturates to $O(n_A)$ in constant time, and there is no extensive complexity drop.

{\bf The generic upper bound.} 
The upper bound for the case of $2m$ disjoint intervals can be thought of as a series in even powers of $\alpha(t) := (2q/(q^2 + 1))^{t-1}$ given by
\begin{equation*}
\Ex_{U_t \sim \nu_t} \big[\tr(\rho_A(U_t)^2) \big] \leq \sum_{\substack{i = 0 \\ \text{$i$ even}}} ^{2m} \vartheta_i\alpha(t)^{i}\,,
\end{equation*}
where the coefficients $\vartheta_i$ are solely functions of the $A$ and $B$ subinterval Hilbert space dimensions. 
Our claim is that $\vartheta_i$ is directly related to the number of non-intersecting chords on an annulus with $2m$ labeled endpoints (interval endpoints) on the outer circle and $i$ unlabeled points (number of non-intersecting walks) on the inner circle. 
The weight of a chord is the Hilbert space dimension of its homologous boundary subsystem, and the weight of a cut (a line segment from the outer to the inner circle) is 1. 
For example, for the four-interval scenario, we have the following power series, which evaluates to the same expression as given above when we convert diagrams to weights via the homology constraint.
\begin{multline}
\label{eq:example_power_series}
    \Ex_{U_t \sim \nu_t}\big[\tr(\rho_A(U_t)^2)\big] \leq \\
    \alpha(t)^0 \times \Bigg[\chords{T}{E}{T}{T} + \chords{T}{E}{N}{T} + \chords{N}{E}{T}{T} + \chords{T}{D}{T}{C} + \chords{T}{D}{N}{C} + \chords{N}{D}{T}{C} \Bigg] \\
    + \alpha(t)^2 \times \Bigg[\chordslines{T}{E}{H}{D}{T} +  \chordslines{T}{D}{V}{T}{C} + \chordslines{T}{T}{H}{C}{E} + \chordslines{T}{C}{V}{E}{D} \Bigg] + \alpha(t)^4 \times \onlycuts \,.
\end{multline}
The diagrams only illustrate the different chord and cut topologies for labeled endpoints. 
Therefore, the diagrams are insensitive to the subinterval sizes. The dependence on the sizes of the subintervals comes from translating the diagrams to weights via the homology constraint.

\begin{lemma}[Diagrammatic upper bound]
\label{lem:main}
Consider an $n$-qudit brickwork random quantum circuit with periodic boundary conditions containing a subsystem $A$ with $m$ disjoint subintervals. 
Then the average purity of $A$ at time $t$, denoted by $\Ex_{U_t \sim \nu_t} \big[\tr(\rho_A(U_t)^2)\big]$, is upper-bounded as follows:
\begin{align}
    \Ex_{U_t \sim \nu_t} \big[\tr(\rho_A(U_t)^2)\big] \leq \sum_{\substack{i = 0 \\ \text{$i$ even} }} ^{2m} \vartheta_i \alpha(t)^i\,,
\end{align}
where $\alpha(t) := (2q / (q^2 + 1))^{t-1}$ and $\vartheta_i$ are weights, given by products of Hilbert space dimensions of intervals, that are determined from homology constraints on the non-intersecting chords of an annulus with $2m$ points on the outer circle and $i$ points on the inner circle.
\end{lemma}

\begin{proof}[Proof of \autoref{lem:main}]
Consider the diagram
\begin{align}
    \label{eq:proof_1}
    \onlycircle\,,
\end{align}
where the qudits reside on the circles and the time increases from 0 at the outer circle to $t$ at the inner circle. 
The black dots on the outer circle denote the endpoints of the $m$ disjoint intervals of $A$ (four, in the case of \eqref{eq:proof_1}). 
If two random walkers collide before time $t$, then we depict that situation on the diagram with a chord between the endpoints. 
If a random walker does not collide with another random walker, then we draw a cut from the endpoint on the outer circle to the corresponding gray dot on the inner circle. 
Since two random walkers cannot pass each other without colliding, we are restricted to diagrams with non-intersecting chords and cuts. 
The weight that the diagram contributes to the sum for the average purity is derived from the upper bound in \autoref{lem:base}, which is obtained by summing the contributions from all collisions between two random walkers situated a certain distance away from one another (from $t=0$ to $t \rightarrow \infty$). 
This upper bound happens to be the inverse Hilbert space dimension of the qudits between the random walkers. 
For periodic boundary conditions, there are two possibilities for the qudit intervals (represented by the solid blue and red lines on the outer circle) and two possibilities for chords between endpoints on the outer circle (represented by the dashed blue and red lines):
\begin{align}
    \onetoone\,.
\end{align}
They are in one-to-one correspondence via the homology constraint, i.e., a chord and a boundary interval correspond to each other when they are homologous.
\end{proof}

\begin{proposition}[Number of terms in $\vartheta_0$]
\label{prop:main_p1}
The number of terms in $\vartheta_0$ is $\binom{2m}{m}$.
\end{proposition}

\begin{proof}
The number of non-intersecting chord configurations on a \emph{circle} with $2m$ marked points is $\frac{1}{m+1}\binom{2m}{m}$. 
In each configuration, the chords divide the interior of the circle into $m+1$ regions. 
Any such diagram can be upgraded to an annulus by placing a puncture, which represents the inner boundary, within any one of these regions. Alternatively, any non-intersecting $m$-chord diagram on the annulus with the homology constraint can be identified with a non-intersecting $m$-chord diagram on the circle with a puncture in one of its $m+1$ regions.
\end{proof}

Having proven the general form of the average purity of subsystems with $m$ disjoint subintervals in \autoref{lem:main} (and \autoref{prop:main_p1}), we specialize that result to subsystems of $m$ equal-size, disjoint subintervals.

\begin{lemma}[Asymptotic average purity for disjoint subintervals of equal size]
\label{lem:equal}
Let $A = \bigcup_i a_i$ and $B = \bigcup_i b_i$, where $\{a_i\}_{i=1} ^{m}$ is a set of $m$ mutually disjoint intervals, each of which is contiguous and of size $n_a = n_A/m$, and $\{b_i\}_{i=1}^{m}$ is a set of $m$ mutually disjoint intervals, each of which is contiguous and of size $n_b = n_B/m$. 
Let the Hilbert space dimension of any $a_i$ and $b_i$ be given by $\da$ and $\db$, respectively, so that $\dA = \da ^{m}$ and $\dB = \db ^{m}$. 
Then the average purity is
\begin{align}
    \Ex_{U_t \sim \nu_t} \big[\tr(\rho_A(U_t)^2)\big] &\leq \sum_{\substack{i = 0 \\ \text{$i$ even}}} ^{2m} \left[\theta_i \alpha(t)^i + O\left(\frac{(\sqrt{\da}\alpha(t))^i}{\dA} \frac{\da}{\db}\right)\right],
\end{align}
where $\theta_i$ correspond to diagrams with chords homologous to as many $a_i$ as possible.
\end{lemma}

\begin{proof}
Fix an even $i$ such that $0 \leq i \leq 2m$. 
There are a total of $m - i/2$ chords, some of which are strictly homologous to $A$ subintervals. 
Any diagram wherein all chords are not strictly homologous to $A$ subintervals will have at least one chord homologous to an interval containing a $B$ subinterval. 
$\theta_i$ counts those diagrams that have $m-i/2$ chords strictly homologous to $A$ subintervals and $i$ cuts. 
There are $m$ $A$ subintervals, of which any $i/2$ intervals can have endpoints on cuts and not chords, so there are $\binom{m}{i/2}$ possibilities. 

The contribution from terms that are not counted in $\theta_i$ is bounded as follows. 
Consider the diagrams with $m-i/2$ chords, of which $j \geq 1$ chords are homologous to intervals containing $B$ subintervals. These chords contribute a weight of at most $\frac{1}{\da^{m - i/2 - j}} \frac{1}{\db^j} = \frac{1}{\da^{m - i/2}} \frac{\da^j}{\db^j}$. 
This proves the order of the correction. 
Therefore,
\begin{align}
    \nonumber
    \Ex_{U_t \sim \nu_t} \big[\tr(\rho_A(U_t)^2)\big] &\leq \sum_{\substack{i = 0\\ \text{$i$ even}}} ^{2m} \left[ \theta_i \alpha(t)^i + O
    \left(\frac{(\sqrt{\da}\alpha(t))^i}{\dA} \frac{\da}{\db}\right) \right] \\
    &= \sum_{\substack{i = 0\\ \text{$i$ even}}} ^{2m} \left[\frac{\binom{m}{i/2}}{\dA} (\sqrt{\da}\alpha(t))^i + O
    \left(\frac{(\sqrt{\da}\alpha(t))^i}{\dA} \frac{\da}{\db}\right) \right], 
\end{align}
where the constant in the error term for each $i \in \{0, 2, \ldots, 2m\}$ is independent of the local dimension $q$ and only depends on the count of the diagrams.
\end{proof}

Finally, using the asymptotic expression for the average purity of subsystems of $m$ equal-size, disjoint subintervals (\autoref{lem:equal}), we derive the timescale after which such subsystems thermalize to the maximally mixed state. 

\begin{theorem}[Thermalization timescale]
\label{thm:main}
After a time
\begin{align}
    t - 1 \geq \max\left\{\left\lceil\frac{1}{\ln\Big(\frac{q^2 + 1}{2q}\Big)} \Big(\frac{n_A}{2m}\ln(q) + \frac{1}{2}\ln\Big(\frac{1}{\ln(1 + 4\delta^2 \epsilon)}\Big) + \frac{1}{2}\ln(m)\Big)\right\rceil,\ \left\lceil \frac{n_A}{2m} \frac{\ln(q)}{\ln\big(\frac{q^2+1}{2q}\big)}\right\rceil\right\}\,,
\end{align}
on at least a $1 - \epsilon$ fraction of random quantum circuits, subsystems with $n_a < n_b$ are $\delta$-close to the maximally mixed state $\iden_A / \dA$.
\end{theorem}

\begin{proof}
First, let us consider $t - 1 \geq \left\lceil \frac{n_A}{2m} \frac{\ln(q)}{\ln\big(\frac{q^2+1}{2q}\big)}\right\rceil$. Then $(\sqrt{\da} \alpha(t))^i < 1$ and
\begin{align}
    \dA \Ex_{U_t \sim \nu_t}\big[\tr(\rho_A(U_t)^2)\big] - 1 &\leq \sum_{\substack{i = 0 \\ \text{$i$ even}}} ^{2m} \left[\binom{m}{i/2} (\sqrt{\da}\alpha(t))^i \right] + O(\da/\db) - 1 \\
    &= (1 + \da \alpha^2 (t))^m - 1 + O(\dA/\dB) \\
    &\leq \exp(\da \alpha^2 (t) m) - 1 + O(\da/\db)\,,
\end{align}
where the constant in the error term $O(\da / \db)$ is independent of the local dimension. 
All the $q$-dependence in the error term is upper-bounded by $\da / \db$. 
It suffices to pick $\da \alpha^2 (t) m \leq \ln(1 \linebreak[1] + \linebreak[1] 4\delta^2 \epsilon)$ because then
\begin{align}
    \Pr_{U_t \sim \nu_t}\Big(\frac{1}{2}\Big\lVert \rho_A(U_t)  - \frac{\iden_A}{\dA} \Big\rVert_1 \geq \delta\Big)
    &\leq \Pr_{U_t \sim \nu_t}\Big(\dA \tr(\rho_A(U_t)^2) - 1 \geq 4\delta^2\Big) \\
    &\leq \frac{\dA \Ex_{U_t \sim \nu_t}\big[\tr(\rho_A(U_t)^2)\big] - 1}{4 \delta^2} \\
    &= \epsilon + O(\da / \db)\,.
\end{align}
To pick $\da \alpha^2 (t) m \leq \ln(1 + 4\delta^2 \epsilon)$, it in turn suffices to pick $t$ such that $t - 1 \geq \left\lceil \frac{n_A}{2m} \frac{\ln(q)}{\ln\big(\frac{q^2+1}{2q}\big)}\right\rceil$ and, simultaneously,
\begin{align}
    t - 1 \geq \left\lceil\frac{1}{\ln\Big(\frac{q^2 + 1}{2q}\Big)} \Big(\frac{n_A}{2m}\ln(q) + \frac{1}{2}\ln\Big(\frac{1}{\ln(1 + 4\delta^2 \epsilon)}\Big) + \frac{1}{2} \ln(m)\Big)\right\rceil.
\end{align}
\end{proof}

Since $\delta$, $\epsilon$, $m$ are independent of the local dimension $q$, \autoref{thm:main} gives $t \geq 
\big\lceil \frac{n_A}{2m} + O\big(\frac{1}{\ln(q)}\big) \big\rceil$ asymptotically in $q$.

\subsection{Far From Thermal Before Thermalization} \label{subsec:qi_far}

For completeness, we prove that the subsystem $\rho_A(U_t)$, for either contiguous or disjoint intervals, is far from maximally mixed before the thermalization time. 
Note that the trace distance between a mixed state $\rho_A (U_t)$ and the maximally mixed state $\iden_A/\dA$ is lower-bounded by the rank of $\rho_A(U_t)$, $r = {\rm rank}(\rho_A(U_t))$. 
To show this, we first note that $\iden_A$ is diagonal in the eigenbasis of $\rho_A(U_t)$, whose eigenvalues we denote by $\lambda_i(U_t)$ for $i = 1, 2, \ldots, \dA$. Then
\begin{equation}
    \Big\|\rho_A(U_t) - \frac{\iden_A}{\dA}\Big\|_1 = \sum_{i=1}^r \Big| \lambda_i - \frac{1}{\dA}\Big| + \frac{\dA-r}{\dA}\geq \sum_{i=1}^r \left(\lambda_i - \frac{1}{\dA}\right) + 1-\frac{r}{\dA} = 2\left(1 - \frac{r}{\dA}\right)\,.
\end{equation}
Thus,
\begin{equation}
    \dist\Big(\rho_A(U_t), \frac{\iden_A}{\dA}\Big) = \frac{1}{2} \Big\|\rho_A(U_t) - \frac{\iden_A}{\dA}\Big\|_1 \geq 1 - \frac{{\rm rank}(\rho_A(U_t))}{\dA}\,.
\end{equation}
The trace distance remains $\geq \delta$ as long as ${\rm rank} (\rho_A(U_t)) \leq \dA (1 - \delta)$. 
Assuming that $\rho_A(U_t)$ evolves via a brickwork random quantum circuit with periodic boundary conditions (i.e., $U_t \sim \nu_t$), after $t$ even layers, the rank of $\rho_A(U_t)$ can increase by at most $q^{2t}$ (regardless of the gate set) when $A$ is contiguous. 
For the same dynamics but for a subsystem composed of $m$ equal-size, disjoint intervals, the rank of $\rho_A(U_t)$ can increase by at most $q^{2t \times m}$. 
Thus, for 
\begin{equation}
    t\leq \left\lfloor\frac{n_A}{2m} - \frac{1}{2m\ln(q)} \ln\Big(\frac{1}{1 - \delta}\Big) \right\rfloor,
\end{equation}
$\rho_A(U_t)$ remains at least $\delta$-far from the maximally mixed state on the same subsystem. 
This does not contradict the fast scrambling theorem of \cite{brown2015decoupling} because there, the dynamics is nonlocal (or all-to-all), whereas here, the dynamics is a local brickwork random quantum circuit. 
More importantly, we use that locality to compute an upper bound on ${\rm rank}(\rho_A(U_t))$.

In $D \geq 2$ dimensions, after $2D$ brickwork layers, accounting for one application of each gate in each dimension, the maximum increase in rank of a subsystem with $m$ disjoint intervals of $(D-1)$-dimensional areas $\partial A_1, \partial A_2, \ldots, \partial A_m$ is $q^{2 \sum_{i = 1} ^m \partial A_i}$. 
After $t$ brickwork layers of $2$-qudit gates, where $t$ is divisible by $2D$, the maximum increase in rank of a subsystem is $q^{2 t \sum_{i = 1} ^m \partial A_i / 2D}$. 
Therefore, for 
\begin{equation}
    t\leq \left\lfloor \frac{D n_A}{\sum_{i = 1} ^m \partial A_i} - \frac{D}{\ln(q) \sum_{i = 1} ^m \partial A_i} \ln\Big(\frac{1}{1 - \delta}\Big) \right\rfloor,
\end{equation}
$\rho_A(U_t)$ remains at least $\delta$-far from the maximally mixed state on the same subsystem.

Taking this and the previous section together, we conclude that in depth-$t$ 1D brickwork random quantum circuits with periodic boundary conditions, subsystems of $m$ equal-size, disjoint intervals are $\delta$-far from maximally mixed up to $t = \big\lfloor \frac{n_A}{2m} - O\big(\frac{1}{\ln(q)}\big) \big\rfloor$ and $\delta$-close to it after $t = \big\lceil \frac{n_A}{2m} + O\big(\frac{1}{\ln(q)}\big)\big\rceil$. 
Therefore, subsystems thermalize (go from far-from-thermal to close-to-thermal) sharply in the limit of large local dimension $q$. 
We have only been able to prove the sharp transition in 1D because in higher dimensions, the detailed description of the purity dynamics in terms of colliding-random-walker configurations becomes more complicated.

\section{From Quantum Mechanics to Holography} \label{sec:qtohol}

In the infinite-temperature limit, we can make a systematic comparison between timescales computed in holography and in random quantum circuits as the number of connected components $m$ varies.  In particular, we can corroborate quantum-mechanical expectations in holography.

First, recall the geometric intuition.  In the simple case of a disconnected subregion consisting of $m$ identical and equidistributed intervals along the circle, holography indicates a sharp transition at $p = 1/2$.\footnote{In random quantum circuits, the transition to exponentially long complexity growth always happens at $p = 1/2$, but for $p < 1/2$, the precise behavior in time depends on whether the subsystem is contiguous.}  When $p < 1/2$ (i.e., when the intervals are shorter than the spaces between them), the minimal static surface is a collection of caps on each interval, and the enclosed volume does not include the ER bridge.  When $p > 1/2$, the minimal static surface is a collection of caps connecting the endpoints of adjacent intervals, and the enclosed volume includes the ER bridge.

To match onto random quantum circuits, we consider the high-temperature ($\beta\to 0$) limit.  In this limit, the total area of the static RT surface is roughly linear in the area of the boundary subregion when $p < 1/2$ and roughly linear in the area of its complement when $p > 1/2$ (where areas are understood as lengths).  Hence the total static area is roughly the same as it would be for $m = 1$.  However, the total area of the growing HM surface is $m$ times larger than in the case $m = 1$ because it has $m$ times as many connected components.  Therefore, increasing $m$ shrinks the regime of $p$ where the growing surface has initially minimal area and modifies the critical subsystem fraction as $p_\text{crit}\to mp_\text{crit}$ (note that $p_\text{crit} = 0$ in the strict infinite-temperature limit).  Moreover, within that regime, the complexity saturation time decreases by a factor of $m$.

On the other hand, the bulk volumes behave in the opposite way with respect to $m$: increasing $m$ does not affect the total volume enclosed by the growing HM surface because it simply splits up the volume into more parts, but it decreases the total volume enclosed by the static RT surface.

We can now state the relevant time and complexity scales more quantitatively.  We start, in the $m = 1$ case, with a two-sided modification of the Haah-Stanford estimates \cite{Haah:2025hyf} for a symmetric subsystem consisting of identical intervals of length $L$ on the left and the right:\footnote{We derive the relevant formulas carefully in Appendix \ref{app:areasandvolumes}, fixing some slight inaccuracies in the formulas quoted by \cite{Haah:2025hyf}.}
\begin{itemize}
\item For $t_b < L/2$, the dominant RT surface crosses the ER bridge.  Ignoring an additive constant, each connected component has length $\frac{\ell}{\pi}\log\cosh\frac{2\pi t_b}{\beta}$, which is approximately $\frac{2\ell t_b}{\beta}$ for small $\beta$.
\item For $t_b > L/2$, the dominant RT surface does not cross the ER bridge.  Ignoring the same additive constant, each connected component has length $\frac{\ell}{\pi}\log\sinh\frac{\pi L}{\beta}$, which is approximately $\frac{\ell L}{\beta}$ for small $\beta$.
\end{itemize}
The timescale of the complexity drop ($t_b = L/2$) is easy to see on the gravity side, as it only requires comparing areas.  Predicting the magnitude of the complexity drop requires calculating volumes; here, we will be more heuristic.\footnote{Haah and Stanford \cite{Haah:2025hyf} define complexity by subtracting the value in the thermal state, and hence set it to zero at late times.  We do not wish to do this.}  For $t_b < L/2$, the volume enclosed by the dominant RT surface is approximately $\frac{2\ell t_b L}{\beta}$ (the product of the horizontal and vertical dimensions).  Let $2L_h$ denote the length of the portion of a single component of the growing surface that lies outside both the left and the right horizons.  See \autoref{fig:wormholes}.  Since each component of the static surface caps off outside the respective horizon, their combined enclosed volume is bounded above by $2L_h L$.  Note that $L_h$ is independent of $L$ and of time.  Therefore, at the thermalization time $t_b = L/2$, the enclosed volume drops from $O(L^2)$ to $O(L)$.  Hence we reproduce, from holography, the $O(n_A^2)$ to $O(n_A)$ complexity drop in random quantum circuits \cite{fan2025sharptransitionssubsystemcomplexity}.

\begin{figure}[!htb]
\centering
\includegraphics[width=0.8\textwidth]{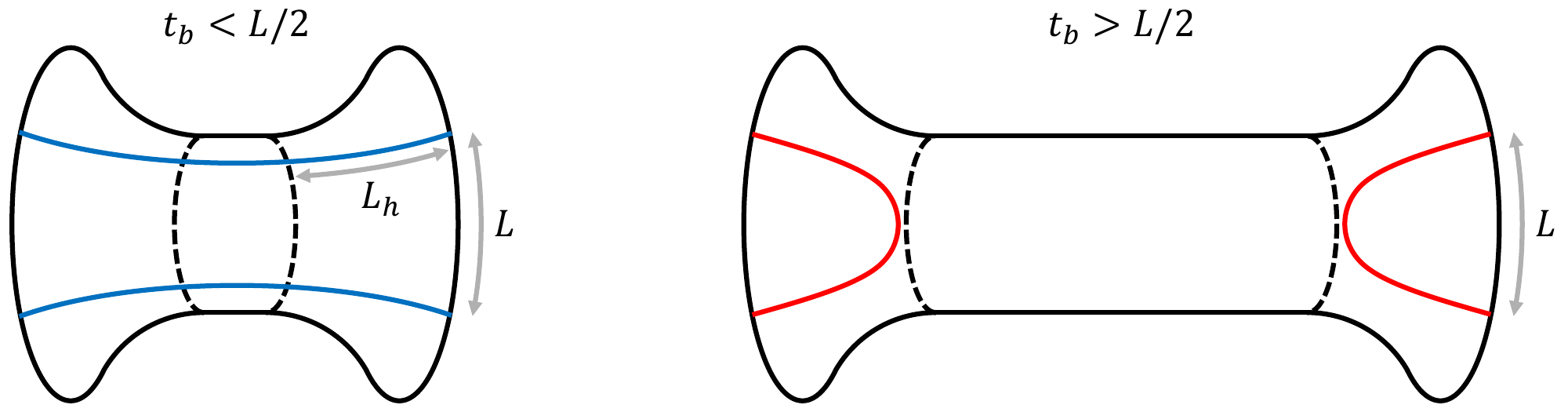}
\caption{Spatial slices of the two-sided geometry before (left) and after (right) the thermalization time.  Horizons are indicated by dashed lines.  Compare to Figure 3 of \cite{Haah:2025hyf}.}
\label{fig:wormholes}
\end{figure}

Now consider $m$ equally spaced intervals of length $L/m$.  In this case, the growing surface has approximate length $2m\cdot \frac{2\ell t_b}{\beta} = \frac{4m\ell t_b}{\beta}$, while the static surface has approximate length $2m\cdot \frac{\ell(L/m)}{\beta} = \frac{2\ell L}{\beta}$.  Hence the transition time is $t_b = L/2m$.  On the other hand:
\begin{itemize}
\item For $t_b < L/2m$, the enclosed volume is approximately $m(L/m)\cdot \frac{2\ell t_b}{\beta} = \frac{2\ell t_b L}{\beta}$.
\item For $t_b > L/2m$, the enclosed volume is bounded above by $m(L/m)\cdot 2L_h = 2L_h L$. (This is a very generous bound, as the enclosed volume should actually decrease with $m$ because the caps extend less far into the bulk with increasing $m$.)
\end{itemize}
Therefore, at the transition, the enclosed volume drops from $O(L^2/m)$ to $O(L)$. This is again consistent with results from random quantum circuits and the discussion at the beginning of \autoref{sec:qi}, specifically the dependence on $m$ proved in \autoref{thm:main}.

The same argument holds in higher dimensions.  In $d$ (boundary spacetime) dimensions, the thermalization time is proportional to the linear size of the subsystem ($\sim n_A^{1/(d - 1)}$), while the area (volume) of the subsystem itself goes like $n_A$.  The complexity just before the transition scales as the product of these two quantities.  Therefore, the complexity drop should go as
\begin{equation}
O(n_A^{1 + 1/(d - 1)}) = O(n_A^{d/(d - 1)})\to O(n_A).
\end{equation}
The value of the complexity before the drop can be reduced by a factor of $m$ if the subsystem is divided evenly into $m$ connected components.  These expectations are reflected in the locality of random quantum circuits.

\section{Discussion} \label{sec:discuss}

In holography, we have demonstrated the counterintuitive finite-temperature effect that a bipartition of a quantum system into noncontiguous subsystems can allow the complexity of the smaller subsystem to grow for an exponentially long time while the complexity of the larger subsystem collapses and equilibrates rapidly.  
We have also demonstrated in both holography and random quantum circuits that dividing a subsystem into multiple disjoint pieces while preserving its total size can reduce the timescale at which its complexity collapses.  
To our knowledge, as of this writing, the former effect has no known quantum-mechanical analogue. 
Clarifying the extent to which this effect is generic to quantum systems at finite temperature, rather than being an artifact of holography, is one of the most important questions left open by our work.

Broadly speaking, the aforementioned counterintuitive effect occurs when we consider subsystems with components that are small compared to the thermal length scale $\beta$.  
Roughly, the mechanism (on a circle) is as follows.  
Let $A$ be a generic noncontiguous subsystem with complement $B$, where $|A| < |B|$.  
By taking the spacing between any two adjacent $A$ intervals to be sufficiently small, we can create a $B$ interval that gives a sufficiently large negative contribution to $S_B$ (via the entropy formula \eqref{eq:f-def}) to make $S_A > S_B$. 
Strictly speaking, we should choose the UV cutoff $\epsilon$ to be small compared to any interval; doing so ensures that all minimal surfaces have positive lengths.

Notice that this procedure only works when $A$ consists of multiple intervals.  
In addition, it only works at finite temperature because taking $\beta\to 0$ prevents one from taking any interval lengths to be small compared to $\beta$.  
Nonetheless, we can consider a different order of limits where instead of taking $\beta\to 0$ first, we also take one of the lengths to 0 in a $\beta$-dependent way.  
In our analysis of the six-interval case, we have seen that this limiting procedure allows us to take the fraction $p$, which governs the relative size of $A$ and $B$, to zero even while the entanglement wedge of $A$ contains the ER bridge (as long as $\rbb > 2/3$ and, surprisingly, at arbitrarily large but finite temperatures).  
Furthermore, we have argued qualitatively that by allowing for arbitrary $m$, we can arrange for $B$ to be arbitrarily large relative to $A$ while the complexity of $A$ grows for an exponentially long time.\footnote{To be clear, we are not conjecturing that arbitrarily small subsystems of large complexity exist in finite-dimen\-sional quantum models. 
It is impossible to see the finite-dimensional analogue of the six-interval high-temperature limit due to the exponentially small fraction in \eqref{eq:squeezed}, but it should be possible to see the analogue of the uberholography example, where it is understood that one halts the recursive hollowing-out procedure when the interval contains too few qubits.}

Our holographic results, while superficially similar to well-known Page-like transitions in time-independent one-sided setups with varying subsystem size, are conceptually distinct:
\begin{itemize}
\item In vacuum AdS, corresponding to the ground state of the quantum system, we reproduce the Page curve: as a function of subsystem size, the entanglement entropy (boundary length) goes up and then down while the complexity (enclosed area) grows monotonically, albeit discontinuously in the case of multiple intervals.  
This picture holds for any number of equidistributed intervals or for a single interval.  
Indeed, for a bipartition of a pure state, the entanglement spectrum is symmetric, but complexity is not.

\item The behavior is different at finite temperature (in a black hole background) due to the homology constraint: both the entanglement entropy and the complexity can increase monotonically with subsystem size.  In this case, the boundary state is thermal and not pure.
\end{itemize}
One benefit of the two-sided black hole is that it allows us to restrict to global pure states on the boundary.

The counterintuitive effects that we find in holography are absent in random quantum circuits. This is because the late-time output states of those circuits are permutation-invariant as an ensemble, and thus late-time subsystem complexities are oblivious to the topology of the subsystems. (Alternatively, random quantum circuits equilibrate to infinite temperature.) We believe that it may be possible to construct examples of random states such that both less-than-half and greater-than-half subsystems have high complexity. However, capturing the exact holographic behavior in finite-dimensional quantum models---large and small subsystems that simultaneously have low and high complexity, respectively---motivates the search for significantly different models for finite-temperature black hole dynamics than random circuits and perhaps even different notions of subsystem complexity altogether.\footnote{For example, prior to the time when random circuits resemble Haar-random unitaries (which are permutation-invariant), definitions of complexity that differ in the geometric locality of the elementary operations may give different answers for subsystem complexity.} 

Besides proving the finite-temperature holographic conjectures in quantum models, it remains to establish the $t = O(\frac{D n_A}{\sum_{i = 1} ^m \partial A_i})$ thermalization timescales for subsystems of $m$ disjoint regions of areas $\{\partial A_i\}$ in $D$-dimensional random quantum circuits. Here, the difficulty resides in determining tight upper bounds for purity dynamics in such circuits, which were available for 1D random circuits via the colliding-random-walks technique of \cite{cotler2022fluctuations}.

\subsection{Mutual Information and Complexity} \label{subsec:hol_mi}

Here, we try to restate what we have observed in terms of mutual information between subintervals. Holographically, a connected entanglement wedge implies that the subintervals have $O(1/G_N)$ mutual information, while a completely disconnected wedge implies zero mutual information among the subintervals.\footnote{All discussion in this subsection only takes into account leading-order effects in $1/G_N$. For example, ``zero mutual information'' means zero at $O(1/G_N)$, ignoring possible subleading corrections.} The latter then means that the subsystem is in a fully separable state of the subintervals.

With this in mind, we start by looking at the four-interval case. There, one observes that a subsystem has exponentially growing complexity if and only if it has a connected entanglement wedge (an entanglement wedge that connects all four intervals $A^{(L)}_1$, $A^{(L)}_2$, $A^{(R)}_1$, $A^{(R)}_2$), provided that no single interval has size greater than or equal to half the system size (for all discussions below, we assume this holds). In other words, a subsystem has exponential complexity if and only if its constituent subintervals have $O(1/G_N)$ mutual information. An equivalent statement is that a subsystem has exponentially growing complexity if and only if the complementary subsystem has a completely disconnected entanglement wedge, i.e., zero mutual information. Of course, in that case, the complexity of the complementary subsystem saturates after a polynomial time.

The role of finite temperature is to allow the smaller subsystem to carry $O(1/G_N)$ mutual information. In the strict infinite-temperature case, the entanglement entropy of an interval is a linear function of its length, so the smaller subsystem is forced to have zero mutual information. However, we have seen that at finite temperature, a smaller subsystem can have $O(1/G_N)$ mutual information if the subinterval sizes are chosen appropriately.

In the six-interval case, one observes that the \textit{smaller} subsystem has exponentially growing complexity if and only if the \textit{smaller} subsystem has a completely connected entanglement wedge, as in \autoref{fig:six}(1).\footnote{For the larger subsystem, the ``only if'' part is not true since the larger subsystem may have exponential complexity without having a completely connected entanglement wedge.} (Recall that the entanglement wedge of the smaller subsystem cannot contain the wormhole when (3), (4), or (5) dominates.) Again, an equivalent statement is that a smaller subsystem has exponentially growing complexity if and only if the complementary subsystem has a completely disconnected entanglement wedge.

For a bipartite system, mutual information is defined as $I(X_1:X_2)=S(X_1)+S(X_2)-S(X_1\cup X_2)$. But for tripartite and general $m$-partite systems, one can define several different quantities that measure different kinds of $m$-partite information. Here, we mention two of them:
\begin{align}
    I_1(X_1:\cdots:X_m)&=\sum_{i=1}^m S(X_i)-S(X_1X_2\cdots X_m),\\
    I_2(X_1:\cdots:X_m)&=\sum_{i=1}^m S(X_i)-\sum_{i<j}^mS(X_iX_j)+\sum_{i<j<k}^m S(X_iX_jX_k)-\cdots +(-1)^{m-1}S(X_1\cdots X_m),
\end{align}
where $X_i X_j=X_i\cup X_j$, etc.
$I_1$ is called multi-information or total correlation. It measures how far the state is from a fully separable state and is always nonnegative.
$I_2$ is typically called $m$-partite information in the holography literature. In general, it does not have a definite sign, but for holographic systems (i.e., systems for which the entanglement entropy is calculated by RT surfaces), this quantity has been proven to be $\leq 0$ for $m=3$ \cite{Hayden_2013}. Note that both quantities reduce to mutual information in the bipartite case.\footnote{One should keep in mind that the mutual information and the quantities $I_1$, $I_2$ do not necessarily capture quantum entanglement between subintervals $X_1,\ldots, X_m$ if the total system $X_1\cdots X_m$ is in a mixed state. This will indeed be the case for us since the total system will be a subsystem $A$ or $B$ of the full system.}
    
How do we characterize a configuration like \autoref{fig:six}(1) in terms of $I_1, I_2$? One can try saying that $A$ has $O(1/G_N)$ $m$-partite information captured by either $I_1$ or $I_2$, but the easiest way to characterize the configuration is to say that $B$ is in a fully separable state:
\begin{equation}
    I_1(B_1^{(L)}:\cdots :B_m^{(L)}:B_1^{(R)}:\cdots :B_m^{(R)})=0.
\end{equation}
Note that $I_2=0$ also holds if $B$ is a fully separable state, but the ``only if'' direction does not hold in general.
Therefore, from what we have observed so far, a smaller subsystem has exponential complexity if and only if the subsystem has $O(1/G_N)$ $m$-partite mutual information, or equivalently the bigger complementary subsystem is in a fully separable state ($I_1=0$).

From here, it is also easy to see why the bigger subsystem can have small complexity. If the bigger subsystem is in a fully separable state of subintervals which are all smaller than half the system size, then each of the subintervals has only polynomial complexity by \cite{fan2025sharptransitionssubsystemcomplexity}, and it should be easy to obtain the total state from these easily preparable subinterval states.

Lastly, we would like to emphasize that these holographic observations are ``too strong'' compared to what we can generally expect.
On the holography side, $O(1/G_N)$ $m$-partite mutual information (in the sense of \autoref{fig:six}(1)) directly implies $O(e^{1/G_N})$ complexity. However, in general, the complexity lower bound from mutual information is very weak (only $O(n)$ where $n$ is the number of qudits in the system, or equivalently $O(1/G_N)$), and cannot account for the exponential complexity.

A similar observation has been made in a different context. The authors of \cite{May_2019, May_2020, May_2022} looked at a certain type of nonlocal task on the boundary, which corresponds to a 2-to-2 scattering process in the bulk. They observed that gravity can perform this task if the initially far-apart systems share $O(\text{system size})$ mutual information, while the boundary quantum system (without the help of gravity) requires at least $O(\exp(\text{system size}))$ entanglement entropy between these two nonlocal regions. These exponential improvements by gravity, both in this setting and in the case of complexity, may suggest a common underlying feature of gravity that remains to be understood.

\subsection{Prospects for Quantum Information}

Quantum verification of finite-temperature holographic conjectures motivates the study of disjoint subsystems of random circuits with conserved charges \cite{rakovszky2018diffusive, khemani2018operator, marvian2022restrictions}. In such random circuits, the ``chemical potential'' corresponding to the conserved charge is expected to play the role of temperature for the limiting subsystem state. Common conserved quantities considered in the literature include $U(1)$ and $SU(q)$ charges. We ask: What are the timescales at which disjoint subsystems of those circuits thermalize? And, more ambitiously: What is their complexity dynamics? 

While we may not have been able to prove finite-temperature holographic predictions in quantum models, we now gather some insights that may be useful for doing so. For verifying those conjectures in quantum state complexity, quantum models may need---in addition to conserved quantities---certain behaviors for the correlation length in their states. This is because simply splitting a less-than-half subsystem into disjoint pieces is not sufficient to witness the counterintuitive holographic predictions. The disjoint pieces cannot be located too far apart, or else the nonlinear nature of RT surfaces will force the larger subsystem to carry the wormhole. This is witnessed by picking a $p$ along the horizontal axis in \autoref{fig:q=1/2} and \autoref{fig:q=1/10} and following a vertical line up: at large enough $\rb$, we are led outside the counterintuitive regime. $U(1)$- and $SU(q)$-invariant states by themselves do not have a notion of locality, as they are permutation-invariant, but local dynamics (e.g., local symmetric random quantum circuits) that preserve those global conserved quantities can restrict states away from permutation invariance \cite{marvian2022restrictions}. We may also expect to demand more than just conserved charges from random circuits so that their late-time states adequately model finite-temperature black hole complexity. One proposition is to consider translationally invariant random circuits, or, in other words, random circuits with $C_n$ conserved charges, where $C_n$ denotes the cyclic group on $n$ letters. Even in that scenario, we expect that all less-than-half subsystems are exponentially close in trace distance to the maximally mixed state and thus of low complexity. However, we expect to witness greater deviation away from the maximally mixed state for subsystems of Haar-random states with versus without translation invariance. Adding further symmetries alongside $C_n$ may prove useful for modeling holographic predictions. Another proposition is to model the late-time state as a matrix product state (MPS) with periodic boundary conditions and with certain transition matrix gap behavior. The gap directly determines the correlation between disjoint regions \cite{svetlichnyy2024matrix}, and a simple ansatz for MPSs can capture entanglement properties of $m$-disjoint subsystems \cite{akhtar2020multiregion}. Finally, we note that the extra need for a certain correlation length in the quantum states may not be necessary if our working notion of quantum complexity is Krylov operator complexity as opposed to state complexity. Witnessing our holographic conjectures for Krylov operator complexity or in features of operator growth remains an exciting open research direction \cite{Caputa:2026ldd}.

\section*{Acknowledgements}

We thank Matteo Ippoliti, Alex May, Robert Myers, Shan-Ming Ruan, Ainesh Sanyal, and Christopher Vairogs for helpful discussions. NHJ and SM acknowledge support in part from DOE grant DE-SC0025615. AK and MK were supported in part by DOE
grant DE-SC0022021 and by a grant from the Simons Foundation (Grant 651678, AK). ChatGPT (5.5 and 5.6 Pro Sol) and Gemini (2.5 Pro) were used to help in generating the TikZ code for Figures \ref{fig:six}, \ref{fig:fractalBTZ}, and \ref{fig:2m-(3)}, suggest improvements to analytics for the infinite-temperature limit of the four- and six-interval cases in Sections \ref{subsec:analyticstemp} and \ref{subsec:analyticssix}, and for proofreading. 

\appendix

\section{Areas and Volumes in Holography} \label{app:areasandvolumes}

In this appendix, we derive some standard geometric quantities in holography to facilitate comparison to random quantum circuits.  We mainly consider the AdS$_3$-Schwarzschild black hole with asymptotic boundary $S^1_\ell\times \mathbb{R}$:
\begin{equation}
ds^2 = -f(r)\, dt^2 + \frac{dr^2}{f(r)} + r^2\, d\theta^2, \qquad f(r) = 1 - \mu + \frac{r^2}{\ell^2}.
\end{equation}
The horizon radius is $r_h = \ell\sqrt{\mu - 1}$, in terms of which $f(r) = (r^2 - r_h^2)/\ell^2$, and
\begin{equation}
\beta = \frac{2\pi\ell^2}{r_h} = \frac{2\pi\ell}{\sqrt{\mu - 1}}.
\end{equation}
Note that we take $\theta\in [0, \pi]$ to be the polar angle rather than the azimuthal angle.  We take as our subsystem the ball (arc) given by $\theta\in [0, \theta_0]$ for some fixed $\theta_0$.

There are three length scales in the problem:
\begin{itemize}
\item $\ell$, the curvature radius of AdS (which equals the radius\footnote{Not the circumference, like in the main text.} of the boundary circle).
\item $L$, the linear size of the subsystem (which equals $2\ell\theta_0$).
\item $\beta$, the temperature of the black hole.
\end{itemize}
Conformal symmetry implies that only two dimensionless ratios of these scales are meaningful.  For example, $\beta/\ell$ characterizes the black hole, while $L/\ell$ characterizes the subsystem size.  In the limit of small $\beta/\ell$ (high temperature and large volume), the black hole background becomes planar.  The Hawking-Page transition occurs at the critical value $\beta/\ell = 2\pi$.

We have $L = 2\ell\theta_0$ and $p = \theta_0/\pi$.  We restrict our attention to the range $p < 1/2$, for which the complexity saturates quickly.  We focus on the regime of small $\beta/\ell$ by taking $\mu\gg 1$; in this planar limit, it is easy to derive scaling laws.  If $\beta/\ell$ is small, then:
\begin{itemize}
\item For $L\gg \beta$, the complexity saturates at roughly the thermalization time $L/2$.
\item For $L\lesssim \beta$, thermalization occurs essentially instantaneously.  We see that $L\sim \beta$ when $\theta_0 \linebreak[1] \sim \linebreak[1] \beta/\ell$.  This relation defines the critical angle at which the saturation time vanishes.
\end{itemize}

\subsection{Static RT Surface Area}

The late-time RT surface is described by a time-independent $\theta(r)$ at fixed $t_b$.  One can read off the asymptotic value $\theta = \theta_0$ for large $r$.  To compute $\theta(r)$, we write the area functional as
\begin{equation}
A[\theta(r)] = 2\int dr\, \sqrt{\frac{1}{f(r)} + r^2\theta'^2}.
\label{areafunctionalstaticr}
\end{equation}
Since the integrand has no explicit $\theta$-dependence, we can write the variational equation in terms of a conserved quantity:
\begin{equation}
\frac{d}{dr}\left(\frac{r^2\theta'}{\sqrt{f(r)^{-1} + r^2\theta'^2}}\right) = 0.
\end{equation}
We impose the initial condition $\theta'(r_0) = \infty$ at $r_0 > r_h$.  Since $f(r)$ is well-behaved at $r = r_0$, as $r\to r_0$ from above, the quantity in parentheses approaches $r_0$.  Therefore, the desired solution is
\begin{equation}
\frac{r^2\theta'}{\sqrt{f(r)^{-1} + r^2\theta'^2}} = r_0 \Longleftrightarrow \theta = \int_{r_0}^r \frac{d\rho}{\rho\sqrt{f(\rho)[(\rho/r_0)^2 - 1]}}.
\end{equation}
Note that the integral is well-behaved at $r = r_0$ even though the integrand diverges, so we indeed have $\theta(r_0) = 0$.  Note also that
\begin{align}
\theta_0 = \theta(\infty) &= \ell r_0\int_{r_0}^\infty \frac{d\rho}{\rho\sqrt{(\rho^2 - r_h^2)(\rho^2 - r_0^2)}} \\
&= \frac{\ell r_0}{2}\int_{r_0^2}^\infty \frac{du}{u\sqrt{(u - r_h^2)(u - r_0^2)}} \\
&= \frac{\ell}{2r_h}\log\left(\frac{r_0 + r_h}{r_0 - r_h}\right).
\end{align}
We can write $r_0$ in terms of $\theta_0$ as follows:
\begin{equation}
r_0 = r_h\coth\frac{r_h\theta_0}{\ell}.
\end{equation}
Letting $r_c$ denote the cutoff radius, we compute the corresponding area:
\begin{align}
A &= 2\int_{r_0}^{r_c} dr\, \sqrt{\frac{1}{f(r)} + r^2\theta'^2} \\
&= 2\ell\int_{r_0}^{r_c} dr\, \frac{r}{\sqrt{(r^2 - r_h^2)(r^2 - r_0^2)}} \\
&= \ell\int_{r_0^2}^{r_c^2} \frac{du}{\sqrt{(u - r_h^2)(u - r_0^2)}} \\
&= 2\ell\log\left(\frac{\sqrt{r_c^2 - r_h^2} + \sqrt{r_c^2 - r_0^2}}{\sqrt{r_0^2 - r_h^2}}\right).
\end{align}
At large $r_c$, we have
\begin{equation}
A = 2\ell\log\left(\frac{2r_c}{\sqrt{r_0^2 - r_h^2}}\right) + O(r_c^{-2}) = 2\ell\log\left(\frac{2r_c}{r_h}\sinh\frac{r_h\theta_0}{\ell}\right) + O(r_c^{-2}).
\end{equation}
The relation between the UV and IR cutoffs is
\begin{equation}
r_c = \frac{\ell^2}{\epsilon},
\label{uvirrelation}
\end{equation}
up to scheme-dependent constant factors that contribute only additive constants to the entropy.  Together, \eqref{uvirrelation} and the Brown-Henneaux relation $c = \frac{3\ell}{2G_N}$ \cite{Brown:1986nw} imply that
\begin{equation}
\frac{r_h\theta_0}{\ell} = \frac{\pi L}{\beta}, \qquad \frac{2r_c}{r_h} = \frac{\beta}{\pi\epsilon}, \qquad \frac{2\ell}{4G_N} = \frac{c}{3},
\end{equation}
and the Ryu-Takayanagi formula $S = \frac{A}{4G_N}$ then gives
\begin{equation}
S = \frac{c}{3}\log\left(\frac{\beta}{\pi\epsilon}\sinh\frac{\pi L}{\beta}\right) + O(\epsilon^2).
\end{equation}

\subsection{Growing HRT/HM Surface Area}

Spherical symmetry guarantees that the growing RT (i.e., HM) surface can be described by $r(t)$ and $\theta(t)$.  For any solution, $r$ ranges from some minimal radius $r_s < r_h$ to $\infty$.

In terms of $r(t)$ and $\theta(t)$, the area functional is
\begin{equation}
A[r(t), \theta(t)] = 2\int dt\, \sqrt{-f(r) + \frac{\dot{r}^2}{f(r)} + r^2\dot{\theta}^2}.
\label{areafunctionalgrowingt}
\end{equation}
Since the integrand has no explicit dependence on $t$ or $\theta$, we obtain two conserved quantities:
\begin{equation}
\frac{d}{dt}\left(\frac{f(r)}{\sqrt{-f(r) + \dot{r}^2/f(r) + r^2\dot{\theta}^2}}\right) = \frac{d}{dt}\left(\frac{r^2\dot{\theta}}{\sqrt{-f(r) + \dot{r}^2/f(r) + r^2\dot{\theta}^2}}\right) = 0.
\end{equation}
We impose the following boundary conditions at the symmetric point $t = 0$:
\begin{equation}
r(0) = r_s, \qquad \dot{r}(0) = 0, \qquad \theta(0) = \theta_s, \qquad \dot{\theta}(0) = 0.
\end{equation}
The two free parameters $r_s$ and $\theta_s$ determine the asymptotic values $t_b$ and $\theta_0$ at $r = \infty$:
\begin{equation}
r(t_b) = \infty, \qquad \theta(t_b) = \theta_0.
\end{equation}
The corresponding solution is given by
\begin{equation}
\frac{f(r)}{\sqrt{-f(r) + \dot{r}^2/f(r)}} = -\sqrt{-f(r_s)}, \qquad \dot{\theta} = 0.
\end{equation}
Note that $f(r_s) < 0$.  Explicitly, we have
\begin{equation}
\dot{r} = -f(r)\sqrt{1 - \frac{f(r)}{f(r_s)}}, \qquad \theta = \theta_s = \theta_0,
\end{equation}
where we choose the sign of $\dot{r}$ as in Appendix B.2 of \cite{fan2025sharptransitionssubsystemcomplexity}.  Note that $f(r) < 0$ for $r < r_h$, while $f(r) > 0$ for $r > r_h$.  We now write
\begin{equation}
t(r) = -\int_{r_s}^r \frac{d\rho}{f(\rho)\sqrt{1 - \frac{f(\rho)}{f(r_s)}}},
\end{equation}
which is understood as a principal value integral about the pole at $r = r_h$.  Since $f(r)$ achieves its minimum at $r = 0$, this integral diverges (i.e., $t_b\to\infty$) as $r_s\to 0$.  Using the indefinite integral
\begin{align}
-\int \frac{d\rho}{f(\rho)\sqrt{1 - \frac{f(\rho)}{f(r_s)}}} &= -\ell^2\sqrt{r_h^2 - r_s^2}\int \frac{d\rho}{(\rho^2 - r_h^2)\sqrt{\rho^2 - r_s^2}} \\
&= \frac{\ell^2}{2r_h}\log\left(\frac{r_h\sqrt{r^2 - r_s^2} + r\sqrt{r_h^2 - r_s^2}}{r_h\sqrt{r^2 - r_s^2} - r\sqrt{r_h^2 - r_s^2}}\right) + C,
\end{align}
we have for $r_s\leq r < r_h$ that
\begin{equation}
t(r) = \frac{\ell^2}{2r_h}\log\left(\frac{r\sqrt{r_h^2 - r_s^2} + r_h\sqrt{r^2 - r_s^2}}{r\sqrt{r_h^2 - r_s^2} - r_h\sqrt{r^2 - r_s^2}}\right)
\end{equation}
and for $r > r_h$ that
\begin{equation}
t(r) = -\left(\int_{r_h + \epsilon}^r + \int_{r_s}^{r_h - \epsilon}\right)\frac{d\rho}{f(\rho)\sqrt{1 - \frac{f(\rho)}{f(r_s)}}} = \frac{\ell^2}{2r_h}\log\left(\frac{r_h\sqrt{r^2 - r_s^2} + r\sqrt{r_h^2 - r_s^2}}{r_h\sqrt{r^2 - r_s^2} - r\sqrt{r_h^2 - r_s^2}}\right),
\end{equation}
where the boundary terms at $r_h\pm \epsilon$ reduce to
\begin{equation}
\mp\frac{\ell^2}{2r_h}\log\left(\frac{2r_h(r_h^2/r_s^2 - 1)}{\epsilon} + O(1)\right),
\end{equation}
which cancel ($\epsilon$ is not to be confused with the CFT UV cutoff).  We read off that
\begin{equation}
t_b = t(\infty) = \frac{\ell^2}{2r_h}\log\left(\frac{r_h + \sqrt{r_h^2 - r_s^2}}{r_h - \sqrt{r_h^2 - r_s^2}}\right),
\end{equation}
which we can invert to obtain
\begin{equation}
r_s = \frac{r_h}{\cosh\frac{r_h t_b}{\ell^2}}.
\end{equation}
We now compute the corresponding area.  Using time reflection symmetry, it can be written as
\begin{align}
A &= 4\int_0^{t_b} dt\, \sqrt{-f(r) + \frac{\dot{r}^2}{f(r)}} \\
&= 4\int_{r_s}^\infty \frac{dr}{|\dot{r}|}\, \sqrt{-f(r) + \frac{\dot{r}^2}{f(r)}} \\
&= 4\int_{r_s}^\infty \frac{dr}{\sqrt{f(r) - f(r_s)}} \\
&= 4\ell\int_{r_s}^\infty \frac{dr}{\sqrt{r^2 - r_s^2}}.
\end{align}
In terms of the cutoff radius $r_c$, we have
\begin{equation}
A = 4\ell\int_{r_s}^{r_c} \frac{dr}{\sqrt{r^2 - r_s^2}} = 2\ell\log\left(\frac{r_c + \sqrt{r_c^2 - r_s^2}}{r_c - \sqrt{r_c^2 - r_s^2}}\right).
\end{equation}
At large $r_c$, this is
\begin{equation}
A = 4\ell\log\left(\frac{2r_c}{r_s}\right) + O(r_c^{-2}) = 4\ell\log\left(\frac{2r_c}{r_h}\cosh\frac{r_h t_b}{\ell^2}\right) + O(r_c^{-2}).
\end{equation}
In terms of CFT parameters, we therefore have
\begin{equation}
S = \frac{2c}{3}\log\left(\frac{\beta}{\pi\epsilon}\cosh\frac{2\pi t_b}{\beta}\right) + O(\epsilon^2).
\end{equation}

\subsection{Area Comparison}

For the sake of comparing areas (lengths), it is convenient to write
\begin{equation}
A_\text{RT} = 2\ell\log\left(\frac{\beta}{\pi\epsilon}\sinh\frac{\pi L}{\beta}\right) + O(\epsilon^2)
\end{equation}
for the area of the static RT surface and
\begin{equation}
A_\text{HM} = 2\ell\log\left(\frac{\beta}{\pi\epsilon}\cosh\frac{2\pi t_b}{\beta}\right) + O(\epsilon^2)
\end{equation}
for \emph{half} the area of the growing RT surface (alternatively, the area of a single connected component of the growing RT surface).  For small $\beta$, these expressions are approximately $2\ell$ times $\pi L/\beta$ and $2\ell$ times $2\pi t_b/\beta$, respectively.  Let us be more precise.  Tracing back through the error analysis and expressing the errors in terms of dimensionless quantities, we have
\begin{align}
A_\text{RT} &= 2\ell\log\left(\frac{2r_c}{r_h}\sinh\frac{r_h\theta_0}{\ell}\right) + O\left(\frac{r_0^2 + r_h^2}{r_c^2}\right) \\
&= 2\ell\log\left(\frac{\beta}{\pi\epsilon}\sinh\frac{\pi L}{\beta}\right) + O\left(\frac{\epsilon^2}{\beta^2}\left(\coth^2\frac{\pi L}{\beta} + 1\right)\right).
\end{align}
Assuming that $\epsilon\ll \beta\ll L$, this is
\begin{equation}
A_\text{RT} = 2\ell\left[\frac{\pi L}{\beta} + \log\frac{\beta}{2\pi\epsilon} + O(e^{-2\pi L/\beta}) + O\left(\frac{\epsilon^2}{\beta^2}\right)\right].
\end{equation}
We also have
\begin{align}
A_\text{HM} &= 2\ell\log\left(\frac{2r_c}{r_h}\cosh\frac{r_h t_b}{\ell^2}\right) + O\left(\frac{r_s^2}{r_c^2}\right) \\
&= 2\ell\log\left(\frac{\beta}{\pi\epsilon}\cosh\frac{2\pi t_b}{\beta}\right) + O\left(\frac{\epsilon^2}{\beta^2\cosh^2\frac{2\pi t_b}{\beta}}\right).
\end{align}
Assuming that $\epsilon\ll \beta\ll L\sim t_b$, this is
\begin{equation}
A_\text{HM} = 2\ell\left[\frac{2\pi t_b}{\beta} + \log\frac{\beta}{2\pi\epsilon} + O(e^{-4\pi t_b/\beta}) + O\left(\frac{\epsilon^2}{\beta^2}\right)\right].
\end{equation}
Hence the transition time at which $A_\text{RT} = A_\text{HM}$ is
\begin{equation}
t_\text{transition} = \frac{L}{2} + O(\beta e^{-2\pi L/\beta}) + O\left(\frac{\epsilon^2}{\beta}\right),
\end{equation}
where we have made the error terms explicit (note that the $\log\frac{\beta}{2\pi\epsilon}$ terms cancel).

\subsection{Volume Comparison}

The static caps are time-independent and lie outside the horizon, so their volumes can be computed on a constant-time slice.  For the static RT surface with area functional \eqref{areafunctionalstaticr}, we fix $t = t_b$ but integrate $\theta$ from 0 to $\theta(r)$, so the induced metric is
\begin{equation}
ds^2 = \frac{dr^2}{f(r)} + r^2\, d\theta^2.
\end{equation}
The corresponding volume functional is
\begin{equation}
V[\theta(r)] = 2\int dr\, \frac{r\theta}{\sqrt{f(r)}}.
\end{equation}
Using the indefinite integral
\begin{equation}
\int \frac{d\rho}{\rho\sqrt{(\rho^2 - r_h^2)(\rho^2 - r_0^2)}} = \frac{1}{2r_h r_0}\log\left(\frac{r_h\sqrt{\rho^2 - r_0^2} + r_0\sqrt{\rho^2 - r_h^2}}{r_h\sqrt{\rho^2 - r_0^2} - r_0\sqrt{\rho^2 - r_h^2}}\right) + C,
\end{equation}
we have the explicit formula
\begin{equation}
\theta = \ell r_0\int_{r_0}^r \frac{d\rho}{\rho\sqrt{(\rho^2 - r_h^2)(\rho^2 - r_0^2)}} = \frac{\ell}{2r_h}\log\left(\frac{r_0\sqrt{r^2 - r_h^2} + r_h\sqrt{r^2 - r_0^2}}{r_0\sqrt{r^2 - r_h^2} - r_h\sqrt{r^2 - r_0^2}}\right).
\end{equation}
Therefore, using a radial cutoff $r_c$, we have
\begin{equation}
V = \frac{\ell^2}{r_h}\int_{r_0}^{r_c} dr\, \frac{r}{\sqrt{r^2 - r_h^2}}\log\left(\frac{r_0\sqrt{r^2 - r_h^2} + r_h\sqrt{r^2 - r_0^2}}{r_0\sqrt{r^2 - r_h^2} - r_h\sqrt{r^2 - r_0^2}}\right).
\end{equation}
Now using the indefinite integral
\begin{align}
\int dr\, \frac{r}{\sqrt{r^2 - r_h^2}}&\log\left(\frac{r_0\sqrt{r^2 - r_h^2} + r_h\sqrt{r^2 - r_0^2}}{r_0\sqrt{r^2 - r_h^2} - r_h\sqrt{r^2 - r_0^2}}\right) \nonumber \\
= \sqrt{r^2 - r_h^2}&\log\left(\frac{r_0\sqrt{r^2 - r_h^2} + r_h\sqrt{r^2 - r_0^2}}{r_0\sqrt{r^2 - r_h^2} - r_h\sqrt{r^2 - r_0^2}}\right) - 2r_h\arccos\frac{r_0}{r} + C,
\end{align}
we obtain
\begin{equation}
V = \ell^2\left[\frac{\sqrt{r_c^2 - r_h^2}}{r_h}\log\left(\frac{r_0\sqrt{r_c^2 - r_h^2} + r_h\sqrt{r_c^2 - r_0^2}}{r_0\sqrt{r_c^2 - r_h^2} - r_h\sqrt{r_c^2 - r_0^2}}\right) - 2\arccos\frac{r_0}{r_c}\right].
\end{equation}
At large $r_c$, we have
\begin{align}
V &= \ell^2\left[\frac{r_c}{r_h}\log\frac{r_0 + r_h}{r_0 - r_h} - \pi + O\left(\frac{r_h}{r_c}\log\frac{r_0 + r_h}{r_0 - r_h}\right) + O\left(\frac{r_0}{r_c}\right)\right] \\
&= \ell^2\left[\frac{L}{\epsilon} - \pi + O\left(\frac{\epsilon L}{\beta^2}\right) + O\left(\frac{\epsilon}{\beta}\coth\frac{\pi L}{\beta}\right)\right],
\end{align}
where we have used that $r_0 > r_h$.  Assuming that $\epsilon\ll \beta\ll L$, we have
\begin{equation}
\coth\frac{\pi L}{\beta} = 1 + O(e^{-2\pi L/\beta})
\end{equation}
and therefore
\begin{equation}
V = \ell^2\left[\frac{L}{\epsilon} + O(1) + O\left(\frac{\epsilon L}{\beta^2}\right)\right],
\end{equation}
including error terms.

For the growing HM surface with area functional \eqref{areafunctionalgrowingt}, the enclosed volume is more subtle to compute.  For the static RT surface, the maximal-volume slice lies at constant time.  However, for the growing HM surface, the CV prescription requires us to find the volume of the maximal-volume slice within the entanglement wedge that ends on the HM surface, and this slice does not lie at constant time.  We do not solve the full embedding problem for the maximal slice, but the qualitative picture is as follows.
\begin{itemize}
\item The maximal-volume slice for the global state approaches a finite radial value $r = r_f$ at late times $t_b$ (except near the boundary), from which one reads off the linear growth factor in the global CV proposal.  In $d = 2$, we have $r_f = r_h/\sqrt{2}$.
\item For a subsystem in $d = 2$, the codimension-two HM surface (which consists of two disjoint segments) approaches $r = 0$ at late times.  Therefore, at late times, the codimension-one maximal-volume slice ending on this surface approaches $r = 0$ at its boundaries.  In between these boundaries, however, the maximal slice dips down to the value $r = r_f$; this is the part that sees linear growth.
\end{itemize}
In $d = 2$, we find that the leading contribution to the late-time volume growth is
\begin{equation}
V = \left(\frac{4\pi^2\ell^2 L}{\beta^2}\right)t_b + O(1),
\end{equation}
where the $O(1)$ term includes both boundary effects and a divergent constant.  The derivation is given below.  Without specifying the regularization scheme, setting $t_b = L/2$ does not allow us to quantify the volume (complexity) jump at the HM transition.

\subsection{Higher Dimensions}

We now consider the AdS$_{d+1}$-Schwarzschild black hole ($d\geq 2$) with asymptotic boundary $S^{d-1}_\ell\times \mathbb{R}$:
\begin{equation}
ds^2 = -f(r)\, dt^2 + \frac{dr^2}{f(r)} + r^2\, d\Omega_{d-1}^2, \qquad f(r) = 1 + \frac{r^2}{\ell^2} - \frac{\mu}{r^{d-2}}.
\end{equation}
The horizon radius $r_h$ is the outermost solution to $f(r_h) = 0$, and
\begin{equation}
\beta = \frac{4\pi\ell^2 r_h}{dr_h^2 + (d - 2)\ell^2}.
\end{equation}
Writing the metric of the unit sphere $S^{d-1}_1$ as $d\Omega_{d-1}^2 = d\theta^2 + \sin^2\theta\, d\Omega_{d-2}^2$ where $\theta\in [0, \pi]$, we consider the ball given by $\theta\in [0, \theta_0]$ for some fixed $\theta_0$.

Again, the planar black hole limit is that of small $\beta/\ell$.  The Hawking-Page transition occurs at the critical value $\beta/\ell = \frac{2\pi}{d - 1}$ ($r_h = \ell$).  Small $\beta/\ell$ means $\mu\gg \ell^{d-2}$, so that
\begin{equation}
r_h\approx (\mu\ell^2)^{1/d}, \qquad \beta\approx \frac{4\pi\ell^2}{d(\mu\ell^2)^{1/d}}.
\end{equation}
For small $\beta/\ell$, the critical angle occurs when $L\sim \beta$.  Letting
\begin{equation}
\omega_{d-1} = \operatorname{vol}(S^{d-1}_1) = \frac{2\pi^{d/2}}{\Gamma(\frac{d}{2})},
\end{equation}
the subsystem fraction is
\begin{equation}
p(\theta_0) = \frac{\omega_{d-2}}{\omega_{d-1}}\int_0^{\theta_0} d\theta\, \sin^{d-2}\theta = \frac{\Gamma(\frac{d}{2})}{\pi^{1/2}\Gamma(\frac{d - 1}{2})}\int_0^{\theta_0} d\theta\, \sin^{d-2}\theta,
\end{equation}
so the subsystem has volume $p(\theta_0)\omega_{d-1}\ell^{d-1}$ and linear size $L\sim p(\theta_0)^{\frac{1}{d - 1}}\ell$.  We take $p < 1/2$.  For $L\gg \beta$, the complexity saturates at the thermalization time, which grows as $p^{1/(d - 1)}$.  For $L\lesssim \beta$, thermalization occurs instantaneously.  We see that $L\sim \beta$ when $p(\theta_0)\sim (\beta/\ell)^{d-1}$, which defines the critical angle.

We first compare the areas of the static and growing RT surfaces.  For both surfaces, we use the large-$r$ asymptotics
\begin{equation}
t(r) = t_b + O(r^{-2}), \qquad \theta(r) = \theta_0 + O(r^{-2}).
\end{equation}
The static RT surface is described by a time-independent $\theta(r)$ that asymptotes to $\theta_0$ at large $r$.  The minimum radius is $r_0 > r_h$.  The static area functional is
\begin{equation}
A[\theta(r)] = \omega_{d-2}\int dr\, (r\sin\theta)^{d-2}\sqrt{\frac{1}{f(r)} + r^2\theta'^2}.
\end{equation}
The growing RT surface can be described by both $t(r)$ and $\theta(r)$, which asymptote to $t_b$ and $\theta_0$ at large $r$, respectively.  The minimum radius is $r_s < r_h$.  The growing area functional is
\begin{equation}
A[t(r), \theta(r)] = \omega_{d-2}\int dr\, (r\sin\theta)^{d-2}\sqrt{\frac{1}{f(r)} - f(r)t'^2 + r^2\theta'^2}.
\end{equation}
In both cases (considering only one side of the growing RT surface), the area has leading divergence
\begin{align}
A &\sim \omega_{d-2}\int^{r_c} dr\, \frac{(r\sin\theta)^{d-2}}{f(r)^{1/2}} \\
&\sim \omega_{d-2}\ell\sin^{d-2}\theta_0\int^{r_c} dr\, r^{d-3} \\
&\sim \frac{\omega_{d-2}\ell r_c^{d-2}\sin^{d-2}\theta_0}{d - 2} \\
&= \frac{p'(\theta_0)\omega_{d-1}\ell^{2d-3}}{(d - 2)\epsilon^{d-2}} \\
&= \frac{\ell^{d-2}}{d - 2}\frac{d\operatorname{vol}(\text{subsystem})/d\theta_0}{\epsilon^{d-2}}.
\end{align}
The only way to see the saturation of complexity is to subtract the leading divergences before comparing the areas.

To compute the volume of the static surface, we fix $t = t_b$ but integrate $\theta$ from 0 to $\theta(r)$, so the induced metric is
\begin{equation}
ds^2 = \frac{dr^2}{f(r)} + r^2(d\theta^2 + \sin^2\theta\, d\Omega_{d-2}^2).
\end{equation}
The corresponding volume functional is
\begin{equation}
V[\theta(r)] = \omega_{d-2}\int dr\, \frac{r^{d-1}}{\sqrt{f(r)}}\int_0^{\theta(r)} d\vartheta\, \sin^{d-2}\vartheta.
\end{equation}
To extract the leading divergence, we use $f(r)\sim r^2/\ell^2$ and $\theta\sim \theta_0$ at large $r$:
\begin{align}
V &\sim \omega_{d-2}\ell\int^{r_c} dr\, r^{d-2}\int_0^{\theta_0} d\theta\, \sin^{d-2}\theta \\
&\sim \frac{p(\theta_0)\omega_{d-1}\ell r_c^{d-1}}{d - 1} \\
&= \frac{p(\theta_0)\omega_{d-1}\ell^{2d - 1}}{(d - 1)\epsilon^{d-1}} \\
&= \frac{\ell^d}{d - 1}\frac{\operatorname{vol}(\text{subsystem})}{\epsilon^{d-1}}.
\end{align}
We can only partially extend the analytical results for areas and volumes to higher dimensions because the RT surfaces cannot be determined analytically in higher dimensions.  In $d > 2$, we only determine the leading divergences in the areas and volumes.

Finally, consider the volume of the growing surface.  We can quantify the late-time linear growth of the maximal-volume slice analytically, although not the divergent part.  First, note that the late-time radius of the global maximal-volume slice and the late-time radius of the HM surface are determined by different variational problems.  The former ($r_f$) is given by \cite{Stanford:2014jda}
\begin{equation}
r_f = \underset{r\in (0, r_h)}{\operatorname{arg\, max}}\left[r^{d-1}\sqrt{-f(r)}\right],
\end{equation}
or, equivalently,
\begin{equation}
2(d - 1)f(r_f) + r_f f'(r_f) = 0.
\end{equation}
On the other hand, the HM area functional contains a transverse factor of $r^{d-2}$, so the corresponding late-time HM radius $r_\text{HM}$ satisfies
\begin{equation}
2(d - 2)f(r_\text{HM}) + r_\text{HM}f'(r_\text{HM}) = 0.
\end{equation}
The distinction between $r_f$ and $r_\text{HM}$ is dramatic in $d = 2$, where $r_f = r_h/\sqrt{2}$ and $r_\text{HM} = 0$.  The fact that the bridge-crossing HM surface approaches $r = 0$ does not imply that the CV volume growth vanishes because most of the volume of the maximal slice for a subsystem is concentrated near the constant interior radius $r_f$.

We can estimate the late-time CV growth rate for a subsystem as follows.  Recall that the full rotationally symmetric maximal-volume slice can be described by an embedding function $r(t)$, and therefore the volume functional
\begin{equation}
V = \omega_{d-1}\int dt\, r^{d-1}\sqrt{-f(r) + \frac{\dot{r}^2}{f(r)}}.
\end{equation}
At late times, we have
\begin{equation}
V\to \omega_{d-1}r_f^{d-1}\sqrt{-f(r_f)}\int dt = 2\omega_{d-1}r_f^{d-1}\sqrt{-f(r_f)}t_b,
\end{equation}
where we have used $t\leftrightarrow -t$ symmetry to evaluate the integral.  For a subregion specified by $\theta_0$ whose entanglement wedge crosses the ER bridge, a similar local argument applies.  Away from its endpoints on the HM surface, the relevant maximal slice is nearly indistinguishable from the maximal slice for the full black hole.  Therefore, we replace $\omega_{d-1}$ by the angular volume of the subregion of $S^{d-1}$ corresponding to the boundary angle $\theta_0$:
\begin{equation}
\omega_{d-1}(\theta_0)\equiv \omega_{d-2}\int_0^{\theta_0} d\theta\, \sin^{d-2}\theta.
\end{equation}
Equivalently, this is the dimensionless volume of the angular cross-section of the bridge included in the entanglement wedge.  We find that
\begin{equation}
V(\theta_0)\to 2\omega_{d-1}(\theta_0)r_f^{d-1}\sqrt{-f(r_f)}t_b + O(1)
\end{equation}
at late times.  The $O(1)$ endpoint effects, the finite-time approach to the linear regime, and the precise volume drop at the HM transition are not fixed by this argument.  Moreover, we have assumed that the linear size $L$ of the subsystem is large compared to $\beta$ so that endpoint effects do not compete with the long-throat contribution to the volume.  We can also estimate the growth rate in the high-temperature, planar black hole limit.  Assuming that
\begin{equation}
f(r)\approx \frac{r^2}{\ell^2} - \frac{\mu}{r^{d-2}}\approx \frac{r^2}{\ell^2}\left(1 - \frac{r_h^d}{r^d}\right),
\end{equation}
maximizing the quantity
\begin{equation}
-r^{2(d - 1)}f(r)\approx \frac{r_h^{2d}}{\ell^2}(x - x^2), \qquad x\equiv \left(\frac{r}{r_h}\right)^d,
\end{equation}
for $r < r_h$ gives
\begin{equation}
r_f\approx \frac{r_h}{2^{1/d}}.
\end{equation}
Our starting assumption is consistent because, indeed,
\begin{equation}
f(r_f)\approx 1 - \frac{r_h^2}{2^{2/d}\ell^2}\approx -\frac{r_h^2}{2^{2/d}\ell^2}
\end{equation}
in the planar limit.  The growth rate is then
\begin{equation}
2\omega_{d-1}(\theta_0)r_f^{d-1}\sqrt{-f(r_f)}\approx \frac{\omega_{d-1}(\theta_0)r_h^d}{\ell}.
\end{equation}
When $d = 2$, the exact growth rate is
\begin{equation}
2\omega_1(\theta_0)r_f\sqrt{-f(r_f)} = \frac{2\theta_0 r_h^2}{\ell} = \frac{4\pi^2\ell^2 L}{\beta^2},
\end{equation}
which is consistent with the approximate result for $d > 2$.

\bibliographystyle{utphys}
\bibliography{refs}

\end{document}